\pdfoutput=1
\documentclass[a4paper,12pt,english]{article}

\usepackage[bindingoffset=0.3cm,textheight=22.5cm,hdivide={2.7cm,*,2.7cm}, vdivide={*,22cm,*}]{geometry}
\usepackage{cite}
\usepackage{youngtab}
\usepackage{amsmath,amsfonts,amssymb,babel,slashed,color,empheq,tensor,amsthm}
\usepackage{tensor}
\usepackage{hyperref}

\makeatletter

\newtheorem{thm}{Theorem}[section]
\newtheorem{lem}[thm]{Lemma}

\newcommand{\dd}{\mathrm{d}}

\newcommand{\del}{\partial}

\def\nn{\nonumber} 
\def\obar{\overline}
\numberwithin{equation}{section}

\def\a{\alpha}  \def\b{\beta}
 \def\g{\gamma} \def\G{\Gamma}
 \def\d{\delta} 

\def\l{\lambda} \def\L{\Lambda}  
    \def\r{\rho}
\def\s{\sigma}  \def\t{\tau}

\def\cA{{\cal A}}  \def\cC{{\cal C}} 
\def\cD{{\cal D}} \def\cE{{\cal E}}  
\def\cG{{\cal G}} \def\cH{{\cal H}}  
 \def\cK{{\cal K}} \def\cL{{\cal L}} 
\def\cM{{\cal M}}   
 \def\cQ{{\cal Q}} \def\cR{{\cal R}} 
 \def\cT{{\cal T}}

\def\R{{\mathbb R}} \def\C{{\mathbb C}} \def\N{{\mathbb N}}

 \def\one{\mbox{1 \kern-.59em {\rm l}}}

\def\msu{\mathfrak{su}}
\def\mso{\mathfrak{so}}

\newcommand{\Tr}{\mathrm{Tr}}

\newcommand{\End}{\mathrm{End}}
\def\hs{\mathfrak{hs}}

\def\Tr{\mbox{Tr}}

\def\und{\underline}

\newcommand{\diag}{\rm diag}

\newcommand{\eq}[1]{(\ref{#1})}

 \allowdisplaybreaks[3]

\begin{document}


\parindent=0cm

\renewcommand{\title}[1]{\vspace{10mm}\noindent{\Large{\bf#1}}\vspace{8mm}} 
\newcommand{\authors}[1]{\noindent{\large #1}\vspace{5mm}}
\newcommand{\address}[1]{{\itshape #1\vspace{2mm}}}


\begin{titlepage}
\begin{flushright}
 UWThPh-2020-5
\end{flushright}
\begin{center}
\title{ {\Large Higher-spin gravity and torsion \\[1ex] 
\ on quantized space-time in matrix models}  }

\vskip 3mm

\authors{Harold C.\ Steinacker{\footnote{harold.steinacker@univie.ac.at}}}

\vskip 3mm

 \address{ 
{\it Faculty of Physics, University of Vienna\\
Boltzmanngasse 5, A-1090 Vienna, Austria  } 
  }

\bigskip

\vskip 1.4cm

\textbf{Abstract}
\vskip 3mm

\begin{minipage}{14.5cm}%
\vskip 3mm

A geometric formalism is developed which allows to describe the non-linear regime of
higher-spin gravity emerging on a cosmological quantum space-time in the IKKT matrix model.
The vacuum solutions are  Ricci-flat up to  an
 effective vacuum energy-momentum tensor quadratic in the torsion,
 which  arises from a Weitzenb\"ock-type higher spin connection.
 Torsion is expected to be significant only at cosmic scales and around very massive objects, and 
 could behave like dark matter. 
 A non-linear equation for the torsion tensor is found, which encodes the Yang-Mills equations of the matrix model.
 The metric and torsion transform covariantly under a higher-spin generalization 
of volume-preserving diffeomorphisms, which arises from the gauge invariance of the matrix model.

\end{minipage}

\end{center}

\end{titlepage}

\tableofcontents
%
%
\section{Introduction}

Our present  understanding of gravity in terms of general relativity (GR)
is incomplete for several reasons. One  problem is that
GR is not renormalizable \cite{Goroff:1985th}, and hence does not define a unique 
quantum theory. This  is reflected 
in  notorious difficulties trying to quantize various formulations of GR.
Broader approaches towards a quantum theory of gravity include notably string theory,
which  leads to gravity in 10 dimensions. However, ad-hoc reductions 
to 3+1 dimensions lead to a lack of predictivity
known as the landscape problem. 

Since gravity is tied to the structure of space-time, a natural strategy is 
to develop a suitable ``quantum'' framework for space-time, based on generalized or noncommutative notions of geometry. However, there is little reason to expect that 
straightforward  attempts to mimic GR in such a framework 
would overcome these issues.
Moreover,  pathological long-distance effects arise generically on non-commutative spaces
due to virtual string-like non-local modes, known as UV/IR mixing  \cite{Minwalla:1999px,Matusis:2000jf,Steinacker:2016nsc}.
They cancel only in the maximally supersymmetric 
IKKT matrix model \cite{Ishibashi:1996xs}, which should thus have the best chance to describe  physics, 
leading to an unexpected link with string theory. 

However, it is not obvious how to obtain gravity  from Yang-Mills-type
matrix models such as the IKKT model. There are intriguing hints such as
non-local gauge transformations and Ricci-flat propagating metric fluctuations
\cite{Rivelles:2002ez,Yang:2006dk,Steinacker:2010rh}, but the presence of an anti-symmetric tensor $\theta^{\mu\nu}$ in 
space-time leads to a dangerous breaking of Lorentz invariance.
This problem can be overcome by considering a higher-spin 
generalization, where  $\theta^{\mu\nu}$ is replaced
by a twisted bundle of such tensors over space-time. This is realized 
 in a simple solution of the 
IKKT model with a mass term interpreted as cosmological FLRW  space-time $\cM$,  based on the doubleton 
representations of $\mso(4,2)$ \cite{Sperling:2019xar,Steinacker:2019awe}, cf. \cite{deMedeiros:2004wb}.
It leads to a higher spin gauge theory 
which is invariant under a higher spin 
generalization of volume-preserving diffeomorphisms, ghost-free at the linearized level, and
includes spin 2 gravitons.
A linearized Schwarzschild-like solution was also found \cite{Steinacker:2019dii}.
The  theory has intriguing structural similarities with 
Vasilievs higher-spin gravity \cite{Vasiliev:1990en,Didenko:2014dwa},
but also crucial differences\footnote{It may also be interesting to compare it with 
chiral higher spin gravity \cite{Ponomarev:2016lrm,Skvortsov:2018jea}, where an action 
can be written in light-cone gauge. An action formulation is also possible for 
conformal higher spin gravity \cite{Tseytlin:2002gz,Segal:2002gd}, notably in 3 dimensions
\cite{Pope:1989vj,Fradkin:1989xt,Grigoriev:2019xmp}.}:
it is defined through an action, both IR and UV scale parameters are 
present, and there are 5 propagating 
metric modes which could be interpreted as ``would-be massive'' gravitons.

In the present paper, we study that higher-spin theory at the non-linear level.
This is not  straightforward because the  model is 
of Yang-Mills type, there is no Einstein-Hilbert action,
and everything is based on Poisson brackets (or commutators).
The  gauge symmetry corresponds to generalized diffeomorphisms rather than local Lorentz 
transformations\footnote{A somewhat related proposal was put forward in \cite{Hanada:2005vr} 
where the matrices are interpreted as covariant derivatives, 
which formally leads to the Einstein equations. However this is not quite borne out here.},
hence it is not some reformulation of GR in the spirit of
MacDowell-Mansouri \cite{MacDowell:1977jt}. 
A more appropriate approach can be found in a paper by 
Langmann and Szabo (LS) \cite{Langmann:2001yr},
who pointed out that a dimensional reduction of 
a  gauge theory in 8-dimensional phase space with 
suitable constraints can be interpreted in terms of 4-dimensional teleparallel gravity through torsion.
Although their setup  does not provide a complete theory, a
similar strategy  provides a geometric understanding of the present model. 
See also \cite{Ho:2015bia,Bonora:2018ggh} for related work.

The first message of the present paper is indeed that torsion, rather than curvature, 
is the key to a geometric understanding of the matrix model. 
The reason is quite simple: the matrices $Z^a$, 
which are the fundamental degrees of freedom of the  model \eq{MM-action},
naturally define a frame $E^a = [Z^a,.]$, which defines the effective metric 
governing all propagating modes. Given this intrinsic frame,
it is natural to consider the associated Weitzenb\"ock connection $\nabla$, 
which is defined by $\nabla E^a = 0$. This connection has no curvature but torsion.
The crucial observation is that this torsion naturally encodes the field strength $[Z^a,Z^b]$ 
of the  model, which underlies the noncommutative gauge theory.
We will obtain a non-linear geometric equation for the torsion, which fully 
captures the underlying Yang-Mills equations of motion of the 
matrix model. This provides a useful description of the resulting higher-spin gravity on $\cM$.  
The Ricci tensor of the effective metric can then 
be computed using the equation of motion for the torsion.
It turns out that torsion leads to a 
specific and very interesting {\em modification} 
of the Einstein equations in vacuum here, in contrast to LS or teleparallel gravity.

At a  more technical level, the present model can be viewed as a reformulation
(rather than a reduction) of a 6-dimensional gauge theory 
to 4-dimensional {\em higher-spin} gravity via torsion.
The 6-dimensional  gauge theory lives on fuzzy $\C P^{1,2}$,
which is a quantized $S^2$ bundle over a 3+1-dimensional FLRW space-time $\cM$.
This is the  geometrical description of the underlying background solution 
as a coadjoint orbit of $SO(4,2)$. The local stabilizer 
group of a point on $\cM$ acts non-trivially on the local $S^2$ fiber,
so that the would-be Kaluza Klein modes become higher spin modes, leading
to  a higher spin theory on $\cM$  \cite{Sperling:2019xar,Steinacker:2016vgf}. 
The concept of torsion allows to fully describe this gauge theory in terms of 
higher-spin gravity. 
The gauge transformations of the underlying matrix model lead to covariant 
transformation laws for the metric and torsion, in terms of generalized (higher-spin) Lie derivatives. 
Although the underlying gauge invariance  is exact, its reformulation in terms of 
4-dimensional geometry is  valid  only in an asymptotic regime. 
Thus gravity and general covariance are understood as emergent phenomena.

The main result of this paper is a closed system of non-linear equations for the metric and the torsion 
in vacuum. Torsion is governed by a Yang-Mills-like equation  \eq{torsion-eq-nonlin-coord}
supplemented by a Bianchi identity \eq{Bianchi-full}, and
the Einstein equation is obtained with an effective energy-momentum tensor due to torsion \eq{Einstein-eq-vac}.
These are equations for higher-spin valued fields\footnote{For the higher spin contributions, the equations
are obtained only in an asymptotic regime, for wavelengths much shorter than the cosmic scale.}.
The cosmological background provides an exact solution  with torsion but without higher spin components,
corresponding to $\omega = -\frac 13$ in vacuum. 
Torsion  generally leads to deviations from Ricci-flatness, however this effect
is typically small in a weak gravity regime, as the energy-momentum tensor for torsion is quadratic. 
Nevertheless, a rough qualitative estimate  suggests that torsion  might lead to an apparent ``dark matter halo''
around very massive objects.

Most importantly, the present model has a non-perturbative definition  
as a Yang-Mills type matrix model, which should  make sense at the quantum level.
There is indeed intriguing evidence from numerical simulations   \cite{Kim:2011cr,Nishimura:2019qal} 
that an expanding 3+1-dimensional 
space-time structure  arises at the non-perturbative level.
It should therefore be possible in principle to test, justify and possibly improve the  analytical
studies with numerical simulations.

This paper considers only the vacuum sector of the theory, and 
the main open issue  is how matter acts as a source of torsion and curvature.
It is clear that the propagation of matter is governed by the effective metric,
and the manifest covariance strongly suggests that matter will lead to the usual source term for the 
Einstein tensor. However extra derivative terms should be expected, 
and quantum effects may be important here.
This needs to be clarified in future work.

The outline of the paper is as follows. We first recall the semi-classical description of the 
space-time and the underlying bundle in terms of Poisson manifolds.
A mathematical formalism is developed in section \ref{sec:bundle-hs},
where the  bundle structure is translated into higher-spin valued fields on 
space-time. Higher-spin valued Lie derivatives are introduced  
in section \ref{sec:gaugeinv-Lie}, which are the key to
covariance under higher-spin valued diffeomorphisms. 
The kinematical setup of the matrix model is translated into this language in section \ref{sec:action-frame}, which
allows to derive the equation of motion for torsion and the Ricci tensor in sections \ref{sec:conservation-law}
and \ref{sec:Ricci-eq}.
To validate these results, a more pedestrian derivation of the latter 
is given in section \ref{sec:alt-Ricci}.
The possible role of torsion as dark matter candidate is briefly discussed.
Some technical considerations are delegated to the appendix, including  the computation of the 
vacuum geometry in section \ref{sec:cosm-BG}.

%
%

%
%
%

\section{Matrix model and cosmological spacetime solution}
\label{sec:MM-semiclass}

Our starting point is the IKKT or IIB matrix model extended by a mass term,
\begin{align}
S[Z,\Psi] = {\rm Tr}\big( [Z^a,Z^b][Z_{a},Z_{b}] + 2 m^2 Z_a Z^a
\,\, + \overline\Psi \Gamma_a[Z^a,\Psi] \big) \ 
\label{MM-action}
\end{align}
where $Z^a \in \End(\cH), \ a=0,...,9$ are hermitian matrices acting on a Hilbert space $\cH$, and 
 $\Psi$ are Majorana-Weyl spinors whose entries are (Grassmann-valued) matrices.
 Indices are contracted with the $SO(9,1)$ -invariant tensor $\eta^{ab}$. 
This model has a manifest $SO(9,1)$ symmetry, and it is invariant under gauge transformations
\begin{align}
 Z^a \to U Z^a U^{-1}
 \label{gaugetrafo-matrix}
\end{align}
by unitary matrices $U \in U(\cH)$.
The model is maximally  supersymmetric for $m^2=0$ which is important for its 
quantization, but we will focus on the bosonic sector here.
The only mathematical structures in the matrix model are matrices 
and commutators, 
which reduce to functions and Poisson brackets in the semi-classical limit.
We must hence learn how to efficiently work with these,  and to cast the system 
into a recognizable geometric form. This depends very much on the background under consideration.
Dropping the fermions,
the model \eq{MM-action} leads to the following equations of motion
\begin{align}
 [Z_b,[Z^b, Z^a]] = m^2 Z^a  \ . 
 \label{matrix-eom}
\end{align}

\paragraph{Doubletons and quantized algebra of functions.}

The basic solution under consideration here is based on special representations of 
$\mso(4,2)$. Let $M^{ab}$ be $\mso(4,2)$ generators, which satisfy
\begin{align}
  [M_{ab},M_{cd}] &= i \left(\eta_{ac}M_{bd} - \eta_{ad}M_{bc} - 
\eta_{bc}M_{ad} + \eta_{bd}M_{ac}\right) \ 
 \label{M-M-relations-noncompact}
\end{align}
with indices $a,b = 0,...,5$.
Now consider the {\em doubleton representations} $\cH_n$ for $n\in\N$, which are 
minimal unitary highest-weight irreps which remain irreducible under $SO(4,1)$.
Then the matrices  
\begin{align}
T^\mu &= \frac 1R M^{\mu 4},  \qquad \qquad \mu=0,...,3 \nn\\
X^\mu &= r M^{\mu 5}, \qquad \qquad X^4 = r M^{4 5}
 \label{X-def}
\end{align}
 satisfy the commutation relations
\begin{subequations}
 \label{basic-CR-H4}
 \begin{align}
  [X^\mu,X^\nu] &= - i \, r^2 M^{\mu\nu}  \eqqcolon i \Theta^{\mu\nu} \,,
  \label{X-X-CR}\\
   [T^\mu,X^\nu] &= \frac{i}{R}\eta^{\mu\nu} X^4 \,,  \label{T-X-CR}\\
[T^\mu, T^\nu] &=  -\frac{i}{r^2 R^2} \Theta^{\mu  \nu} \, . \label{T-T-CR} 
%
\end{align}
\end{subequations}
It turns out that the operator algebra $End(\cH_n)$ can be viewed as 
quantized algebra of functions on a $S^2$ bundle over some 3+1-dimensional space-time $\cM$.
To see this, consider the semi-classical limit  $n \to \infty$ indicated by $\sim$, 
where $End(\cH_n)$
becomes a commutative algebra $\cC$ of functions generated by  
$x^\mu \sim X^\mu$ and $t^\mu \sim T^\mu$ , with Poisson brackets $\{.,.\} \sim -i [.,.]$ 
arising from the commutation relations.
For the doubleton representations under consideration, these generators satisfy 
additional relations, which in the semi-classical limit reduce to
\begin{subequations}
\label{geometry-H-M}
\begin{align}
 x_\mu x^\mu &= -R^2 - x_4^2 = -R^2 \cosh^2(\eta) \, , 
 \qquad R \sim \frac{r}{2}n   \label{radial-constraint}\\
 t_{\mu} t^{\mu}  &=  r^{-2}\, \cosh^2(\eta) \,, \\
 t_\mu x^\mu &= 0, \label{orth-xt} \\
 t_\mu \theta^{\mu\a} &= - \sinh(\eta) x^\a , \\
 x_\mu \theta^{\mu\a} &= - r^2 R^2 \sinh(\eta) t^\a , \label{x-theta-contract}\\
 \eta_{\mu\nu}\theta^{\mu\a} \theta^{\nu\b} &= R^2 r^2 \eta^{\a\b} - R^2 r^4 
t^\a t^\b + r^2 x^\a x^\b  
%
\end{align}
\end{subequations}
where $\mu, \a = 0,\ldots ,3$. One can also show that $\theta^{\mu\nu}$ can be expressed
as \cite{Sperling:2019xar}
\begin{align}
 \theta^{\mu\nu} &= \frac{r^2}{\cosh^2(\eta)} 
   \Big(\sinh(\eta) (x^\mu t^\nu - x^\nu t^\mu) +  \epsilon^{\mu\nu\a\b} x_\a t_\b \Big) \ .
   \label{theta-general}
\end{align}
Thus $X^\mu$ can  be interpreted as quantized  functions 
\begin{align}
 X^\mu \sim x^\mu: \quad \cM \ \hookrightarrow \ \R^{3,1}, \qquad  \ \mu = 0,...,3 \ ,
\end{align}
onto the region $-x_\mu x^\mu \geq R^2$.
$\cM$ turns out to be a two-sheeted cosmological FLRW space-time with manifest $SO(3,1)$ symmetry.
\begin{figure}
\begin{center}
 \includegraphics[width=0.5\textwidth]{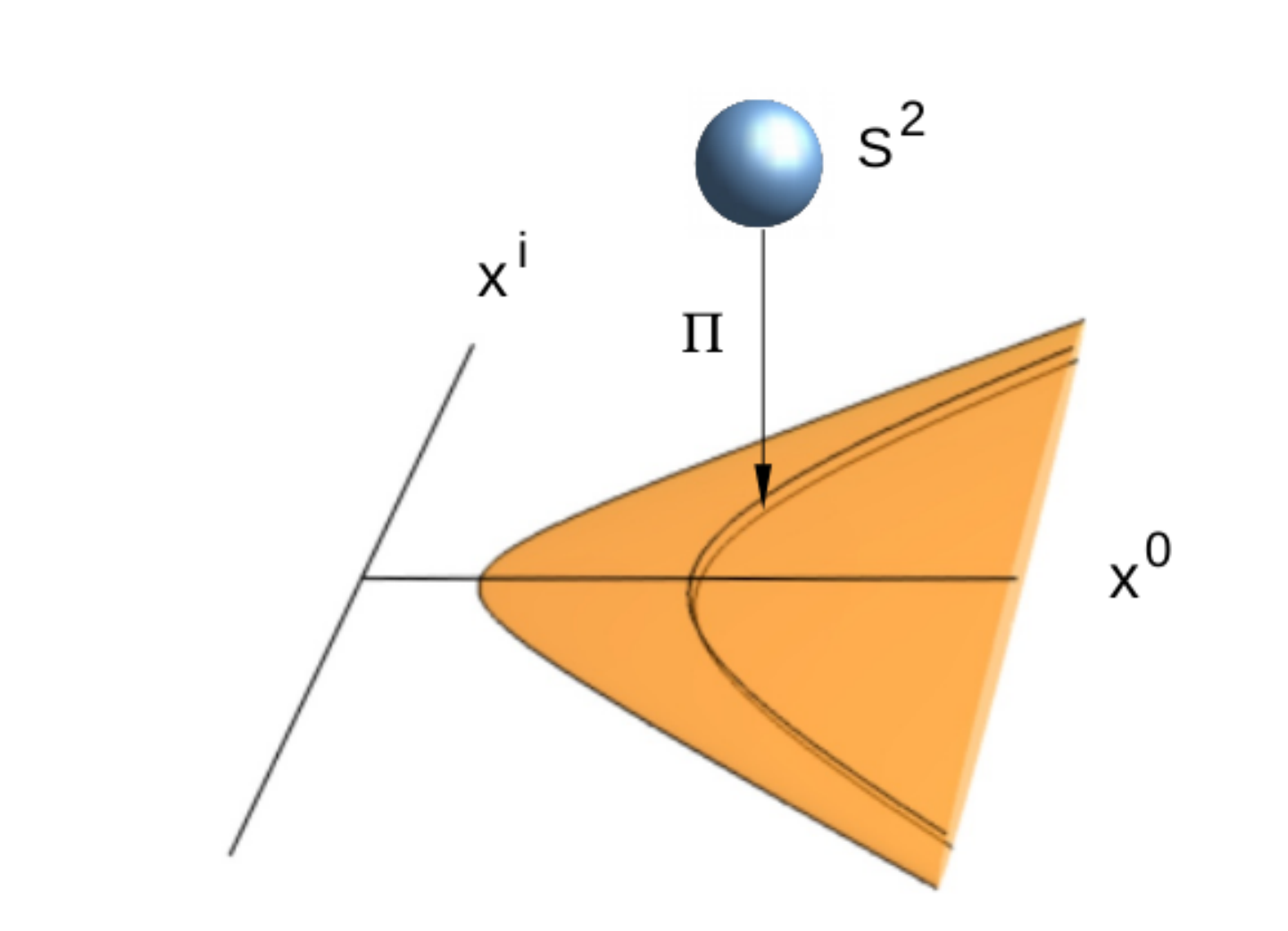}
 \caption{$S^2$ bundle over $\cM$ with bundle projection $\Pi$.}
\end{center}
 \label{fig:projection}
\end{figure}
Here $\eta$ is a global time coordinate defined by
\begin{align}
 R \sinh(\eta) = \pm \sqrt{-x_\mu x^\mu-R^2} \ = x^4 \ ,
 \label{x4-eta-def}
\end{align}
which is related to the scale parameter of the universe \cite{Steinacker:2019dii}
\begin{align}
 a(t)^2 =  R^2 \cosh^{2}(\eta) \sinh(\eta) \ .
 \label{cosm-scale}
\end{align}
The sign of $\eta$ separates the two sheets of $\cM$,
 with a big bounce at $\eta = 0$. Similarly, the $t^\mu$
are extra generators which describe the internal $S^2$ fiber over every point on $\cM$, see figure \ref{fig:projection}.
This $S^2$ is space-like due to \eq{orth-xt} with radius $r^{-2}\cosh^2(\eta)$.
Together, $x^\mu$ and $t^\mu$ generate the algebra $\cC \cong \cC^\infty(\C P^{1,2})$ 
of functions on a 6-dimensional bundle, which 
turns out to be $\C P^{1,2}$. This is an 
$SO(4,1)$- equivariant bundle over $H^4$, which projects onto $\cM$. 
As explained in \cite{Sperling:2019xar,Steinacker:2019fcb} there is an equivariant quantization map
\begin{align}
 \cQ: \quad \cC =  \cC^\infty(\C P^{1,2}) &\to End(\cH_n)  
 \label{quantization-map}
\end{align}
which allows to identify operators with functions up to some cutoff. 
This paper is devoted to  the semi-classical limit, 
replacing the rhs with the lhs, and 
commutators by Poisson brackets.
In particular,
the relation \eqref{T-X-CR} implies that
the derivations
\begin{align}
 -i[T^\mu,.] \sim \{t^\mu,.\}  \ = \sinh(\eta) \del_\mu
\end{align}
act as momentum generators on $\cM$,
leading to  the useful relation 
\begin{align}
 \del_\mu \phi = \frac{1}{\sinh(\eta)} \{t_\mu,\phi\}
 \label{del-t-rel}
\end{align}
for $\phi = \phi(x)$.

\paragraph{Higher spin sectors and Poisson brackets.}

By expanding the algebra $\cC$ into polynomials of minimal degree $s$ in $t^\mu$, 
we obtain a decomposition
\begin{align}
 \cC = \cC^\infty(\C P^{1,2}) = \bigoplus\limits_{s=0}^\infty \cC^s
\end{align}
 into spin $s$ sectors $\cC^s$. These sectors can also be defined in terms of a $\mso(4,2)$ Casimir \cite{Sperling:2019xar}. Here $\cC^0$ are functions of $x$,
 and $\cC^s \ni \phi_{\und\a}(x) t^{\und\a}$ where $t^{\und\a}\equiv t^{\a_1 ... \a_s}$.
 The projection of $\phi\in\cC$ to 
 $\cC^s$ will be denoted by $[\phi]_s$ or $\phi^{(s)}$.
 The constraints \eq{geometry-H-M} reveal that
the multiplication respects this grading as follows:
\begin{align}
 \cC^s \cdot\cC^{s'} \ \in \ \cC^{s+s'} \oplus \cC^{s+s'-2} \oplus ... \oplus \cC^{|s-s'|} \ .
\end{align}
In particular, each $\cC^s$ is a $\cC^0$ module, so that
$\cC$ can be viewed as space of sections of some vector bundle over $\cM$. 
This is a  module of higher-spin fields,
and the bundle structure is  encoded in the sub-algebra $\cC^0\subset\cC$.
The Poisson structure respects this as follows:
\begin{align}
 \{t,\cC^{s}\} &\in \cC^{s} \nn\\
 \{x,\cC^{s}\} &\in \cC^{s+1} \oplus  \cC^{s-1}
\end{align}
and the general structure follows from the derivation property:
\begin{align}
  \{\phi_\a(x) t^\a,\cC^{s}\} &=  \phi_\a(x) \{t^\a,\cC^{s}\}  + t^\a\{\phi_\a(x) ,\cC^{s}\}  
   \qquad \in \ \cC^{s+2} \oplus  \cC^{s} \oplus  \cC^{s-2} \ .
\end{align}
This is a rather complicated structure, but we will argue in the following
that typically the terms involving
space-time derivatives $\del$
are dominant, except possibly in the extreme IR.

\paragraph{Background solution and linearized fluctuations.}

The above geometry provides the geometric interpretation of the following 
{\em  solution} of \eq{matrix-eom} \cite{Sperling:2019xar}
\begin{align}
 \bar Z^\mu = T^\mu = \frac 1R  M^{\mu 4}, \qquad \mu =0,...,3 \qquad \mbox{for} \ \ R^{-2} = \frac 13 m^2  
 \label{T-solution}
\end{align}
with all remaining matrices $\bar Z^a, \ a=4,...,9$ set to zero.
This is the cosmic background under consideration. We will also consider more general solutions
$Z^\mu\in End(\cH_n)$ of \eq{matrix-eom},
which can be viewed as deformations thereof.
Any such background defines a matrix d'Alembertian 
\begin{align}
 \Box := [Z^\mu,[Z_\mu,.]] \  \sim \ -\{Z^\mu,\{Z_\mu,.\}\}  \ 
 \label{Box-def}
\end{align}
which acts on   $\phi \in  \End(\cH_n) \sim \cC^\infty(\C P^{1,2})$, and
will play a  central role in the following. 
 In the semi-classical limit,
we can consider these $Z^\mu$ as elements of the Poisson algebra $\cC$ using the correspondence \eq{quantization-map},  
i.e. as functions on the cosmic background.

On any such background $Z^\mu$, the matrix model defines an action for the fluctuations 
\begin{align}
 Z^\mu \to Z^\mu + \cA^\mu, \qquad \cA^\mu  \in \cC \ .
\end{align}
We briefly recall the main results 
for these fluctuation modes $\cA$, which
transform under a higher-spin-type gauge invariance inherited from \eq{gaugetrafo-matrix}.
The spectrum of square-integrable on-and off-shell modes  $\cA$ on $\bar Z^\mu$ was fully classified in \cite{Steinacker:2019awe}.
It turns out that off-shell, there are 4 towers of higher spin modes for any spin $s> 0$, and 2 towers for $s=0$.
Among these, there are two towers of physical (=gauge-fixed on-shell) higher spin modes for any spin $s\geq 1$.
These towers are truncated with maximal spin\footnote{There 
is a subtlety due to the fact that only the space-like local rotations are manifest on the FLRW background
but boosts are not.
With ``spin $s$ modes'' we refer to the 4-dimensional spin from the underlying $H^4$ point of view, which contain  $2s+1$ degrees of freedom   
corresponding to massive modes. However the mass is essentially zero, so that they decompose further under the FLRW isometry group $SO(3,1)$
into different massless (transverse traceless) modes. For example, the 5 dof of a spin 2 mode decomposes into 
2 dof of a massless graviton, 2 dof of a massless vector field and a scalar.
For further details we refer to \cite{Steinacker:2019awe,Sperling:2019xar}.} $n$. In the semi-classical limit $n\to\infty$, 
the truncation disappears, and all physical modes are shown to have  positive-definite inner product. 

In the remainder of this paper, we will develop a geometric description of the vacuum solutions in the non-linear regime,
avoiding the split into background and fluctuations.
The detailed relation with the above  modes should be clarified in future work.

\subsection{Bundle structure and higher-spin fields as push-forwards}
\label{sec:bundle-hs}

The crucial geometrical observation is as follows. The bundle projection $\Pi:\C P^{1,2} \to \cM$
(see figure \ref{fig:projection})
allows to map  vectors from $\C P^{1,2}$ to $\cM$ via the push-forward,
but it does  not  map  vector fields on $\C P^{1,2}$ to vector fields on $\cM$
because $\Pi$ is not injective. 
However, one can make sense of this as {\em push-forward of vector fields on $\C P^{1,2}$
to $\hs$-valued vector fields on $\cM$},
\begin{align}
 \Pi_*: \quad \Gamma \C P^{1,2} \  \to  \ \Gamma \cM \otimes \hs \ 
 \label{projection-hs}
\end{align}
and similarly for tensor products of vector fields.
Here $\hs$ is the space of functions on the fiber $S^2$, which is spanned by 
polynomials in the $t^\mu$.  The algebra $\cC$ of
functions  on $\C P^{1,2}$ is viewed as algebra of $\hs$-valued functions on $\cM$.
This is the crucial concept  which will allow
to translate the gauge theory on $\C P^{1,2}$ as a higher-spin gauge theory on $\cM$,
which can be viewed as generalized gravity theory\footnote{This is the crucial step beyond the work \cite{Langmann:2001yr}.}. 
Note that the map \eq{projection-hs} respects\footnote{This is the structure of an equivariant bundle, which is in contrast to
 standard Kaluza-Klein reduction, where the 
stabilizer group of a point on the base does not act on the fiber.}  $SO(3,1)$, which implies that the modes in $\hs$
translate into higher-spin modes on  $\cM$. 
However in contrast to Vasiliev's higher spin gravity, the algebra generated by 
the $t^\mu$ is commutative (in the semi-classical limit under consideration), 
and  their nonabelian structure is encoded in their 
Poisson brackets $\{t^\mu,t^\nu\} = -\frac{1}{r^2 R^2} \theta^{\mu\nu}$ \eq{T-T-CR}.
As a vector space, the constraints including \eq{geometry-H-M} 
encode the degrees of freedom\footnote{The relation with e.g. $\hs(so(3,2))$ is more
visible using the Euclidean point of view on fuzzy $H^4_n$. The present Minkowski version 
coincides as a vector space but has modified covariance properties.} 
a $\hs$-algebra in 4 dimensions at each point on $\cM$;
for more details see e.g. section 4 in \cite{Sperling:2018xrm}.

There are hence 2 different points of view of the present structure.
One can take the 6-dimensional picture of $\C P^{1,2}$ as a 6-dimensional manifold,
which is a bundle over $\cM$.  
However from the physical point of view, it is better to view this structure 
in terms of $\hs$ -valued functions and tensors on $\cM$.
We will use both points of view in this paper as appropriate, but emphasize the  more 
physics-oriented 3+1-dimensional picture.

\paragraph{Poisson structure and Hamiltonian vector fields.}

As usual,
the Poisson structure on $\cC$ allows to associate to any $\L\in\cC$ a vector field on $\C P^{1,2}$, via
\begin{align}
 \{\L, \phi \} \equiv \cL_{\xi} \phi, 
       \qquad \L = \L^* \in \cC \ .
 \label{Hamiltonian-1}
\end{align}
This is the Lie derivative of $\phi \in \cC^\infty(\C P^{1,2})$ along
the ``Hamiltonian'' vector field $\xi = \{\L,.\}$, which defines a (one-parameter family of) 
diffeomorphism  on $\C P^{1,2}$, more precisely a symplectomorphism.
Restricted to $\cC^0\subset\cC$, this takes the familiar form
\begin{align}
 \d_\L \phi^{(0)} &= \cL_{\xi} \phi^{(0)} = \xi^\mu \del_\mu\phi ^{(0)}
 \label{gauge-vectorfield}
\end{align}
where
\begin{align}
 \xi &= \xi^\mu \del_\mu, \qquad  \xi^\mu = \{\L,x^\mu\}\ .
\end{align}
This is nothing but the push-forward  $\Pi_*\xi$  as in \eq{projection-hs}, and we will use the same symbol $\xi$
for both pictures. Clearly the components $\xi^\mu$  have in general components in different $\cC^s$ sectors,
which means that $\xi$ is viewed as $\hs$-valued vector field on $\cM$.

Similarly, we can also consider the push-forward of the 
Poisson structure or bivector field $\{.,.\}$ on $\C P^{1,2}$ as $\hs$-valued Poisson structure on $\cM$,
denoted as
\begin{align}
 \{\phi,\psi\}_\cM = \theta^{\mu\nu}\del_\mu\phi \del_\nu\psi \ .
 \label{Poisson-M-def}
\end{align}
Note that $\theta^{\mu\nu} = \{x^\mu,x^\nu\}$ takes indeed values in $\cC^1$, and any coordinates on $\cM$ can be used here.
This respects the Jacobi identity as long as it acts on $\cC^0$, but not in general since the push-forward kills vertical vector fields.
The push-forward basically means that $\hs$  (or the fiber coordinates) is considered as commutative.
We will see that replacing $\{.,.\}$ by $\{.,.\}_\cM$ is a good approximation 
for functions whose wavelengths are much shorter than cosmic scales and for low spin, 
since then the Jacobi identity holds to a very good approximation. 
This will be justified in the next sections.
Then $\xi = \{\L,.\}_\cM$ defines a $\hs$-valued Hamiltonian vector field on $\cM$, and we can drop the subscript
$\cM$ in the appropriate  regime.
No information is lost in this step, due to the following Lemma:

\begin{lem}

Hamiltonian vector fields $\{\L,.\}$ on $\cC$ are uniquely determined by their action on $\cC^0$, i.e. if
 $\{\L,x^\mu\} = 0 \ \forall \mu$ (in some open neighborhood in $\cM$), then $\{\L,.\} \equiv 0$
(in the open neighborhood in $\cM$).
\label{lemma:uniqueness-VF-C0}
\end{lem}

\begin{proof}
This follows from the Jacobi identity, since
\begin{align}
 \{\L,\theta^{\mu\nu}\} = \{\L,\{x^\mu,x^\nu\}\} = -\{x^\mu,\{x^\nu,\L\}\}  - \{x^\nu,\{\L,x^\mu\}\}   = 0 \nn
\end{align}
and $\theta^{\mu\nu}$ together with $x^\mu$ generate the full algebra $\cC$.

\end{proof}
The following two sections provide a more detailed justification for a certain approximation,
which is used later to obtain the reduced 3+1-dimensional equations.
Although this is very important, one may  skip these sections at first reading
and jump  to section \ref{sec:gaugeinv-Lie}, where the higher-spin gauge invariance 
is discussed.

\subsection{Scales and orders of magnitude}

Using $SO(3,1)$ invariance, we can restrict ourselves to the ``reference point''
$p = (x^0,0,0,0)$ on $\cM$. Then \eq{theta-general} reduces to 
\begin{align}
  \theta^{0i} &\stackrel{p}{=}  \frac{r^2}{\cosh^2(\eta)} \sinh(\eta) x^0 t^i 
   \sim r^2 R t^i \nn\\
  \theta^{ij} &\stackrel{p}{=} \frac{r^2}{\cosh^2(\eta)} x^0 \epsilon^{0ijk} t_k 
  \sim \frac 1{\sinh(\eta)} r^2 R \epsilon^{ijk} t^k 
 \label{CR-explicit-ref}
\end{align}
so that $ \theta^{0i} \sim r^2 R t^i  \gg \theta^{ij} $  at late times $\eta \gg 1$,
indicated by $\sim$.
Then the first term in \eq{theta-general} dominates, and
\begin{align}
 \theta^{\mu\nu}
   \sim \frac{r^2}{\cosh(\eta)} (x^\mu t^\nu - x^\nu t^\mu) 
  \label{theta-approx} \,  .
\end{align}
Note that  at $p$, the $t^\mu$ generators satisfy 
\begin{align}
 \{t^i,t^j\} &\stackrel{p}{=} -\frac{1}{r^2 R^2} \theta^{ij} = -\frac 1{R\sinh(\eta)} \epsilon^{ijk} t^k 
 \label{t-t-CR}
\end{align}
and $t^0\stackrel{p}{=}  0$.
Hence in a sense we are considering  a $U(\msu(2))$-valued  Yang-Mills-type theory of local space-like translations, 
however the non-commutativity (via the Poisson bracket) leads to novel structures 
not usually encountered in classical Yang-Mills theory or gravity. 
Note that the explicit $\epsilon^{ijk}$ in \eq{t-t-CR} implies that 
parity is broken\footnote{I would like to thank the anonymous referee for pointing this out.}, and 
the sign  would be flipped
for the doubleton representation $\cH'_{n}$ with opposite chirality. 
This is somewhat reminiscent of chiral higher spin theory \cite{Ponomarev:2016lrm}.

\paragraph{Orders of magnitude and asymptotic regime.}

Now consider the ``size'' of the various generators. The relations \eq{geometry-H-M} lead to 
the following scale estimates
\begin{align}
 |t|   &\sim r^{-1} \cosh(\eta)  \nn\\
 |\theta^{\mu\nu}|  &\sim R r \cosh(\eta) =: L_{\rm NC}^2 \ 
 \label{L-NC-def}
\end{align} 
defining the scale of noncommutativity $L_{\rm NC}$.
This gives for $\phi^{(1)} = \phi_\a(x) t^\a \ \in \cC^1$
\begin{align}
 \{t^\mu,\phi^{(1)}\} &= \{t^\mu,\phi_\a(x) \}t^\a + \phi_\a\{t^\mu, t^\a\} \
  \sim \sinh(\eta)\del\phi_\a t^\a -  \frac{1}{r^2R^2}\phi_\a \theta^{\mu\a} \nn\\
  &=  \frac 1{r R} O\Big(\sinh(\eta)\big(\sinh(\eta) R\del\phi  +  \phi\big)\Big) \nn\\
  &\sim \{t^\mu,\phi(x)_\a \}t^\a \ \big(1+O(\frac{1}{|x|\del})\big)
\end{align}
where $|x|= R\cosh(\eta)$ measures the cosmic time, and $\sim$ indicates $\eta \gg 1$. 
Hence  the derivative term dominates except for extremely low wavelengths $\l$, i.e. in the  {\bf asymptotic regime}
\begin{align}
 L_{\rm NC} \ll \l \ll R\cosh(\eta)
 \label{asymptotic-scale}
\end{align}
much shorter than the cosmic scale  and  longer than the scale of noncommutativity.
Similarly,
\begin{align}
 \{x^\mu,\phi^{(1)}\} &= \{x^\mu,\phi_\a \}t^\a + \phi_\a\{x^\mu, t^\a\} \
  = \theta^{\mu\nu}t^\a \del_\nu\phi_\a  + \sinh(\eta)\phi_\mu \nn\\
  &\sim \sinh(\eta)\big(\sinh(\eta) R\del\phi  + \phi\big) \nn\\
 &\sim \{x^\mu,\phi(x)_\a\} t^\a  \ \big(1+O(\frac{1}{|x|\del})\big) \ .
\end{align}
Note that the estimate applies equally to the component in $\cC^0$ and $\cC^2$, since 
$\theta^{\mu\nu}t^\a$ and $[\theta^{\mu\nu}t^\a]_0$
are comparable in size.
Hence 
in general we can write 
\begin{align}
\{\phi,\psi\}  &\sim \{\phi_{\und\a} ,\psi_{\und\b}\} t^{\und\a} t^{\und\b} \ ,
\label{generic-commutator-estimate}
\end{align}
where $\und\a$ is a multi-index. 
This amounts precisely to the projection  to $\cM$ \eq{Poisson-M-def}
as discussed above, which will be very helpful to extract the leading contributions for physics. 
Recalling  the cosmic FLRW scale parameter \eq{cosm-scale}, the ratio of IR and UV scales is
\begin{align}
 \frac{a(t)^2}{L^2_{NC}} \sim  \frac{R\cosh^2(\eta)}{r} \ \to \ \infty
\end{align}
at late times.
Hence space grows indeed much faster than the NC scale, as it should.

\subsection{Poisson bracket and reduction to $\cM$}
\label{sec:poisson-red}

To make the reduction to $\cM$ precise, we derive an explicit
formula for the Poisson structure with the form
\begin{align}
 \{f,g\} &=  \{f,x^\mu\}M_\mu^\nu\{t_\nu,g\} - \{g,x^\mu\}M_\mu^\nu\{t_\nu,f\} \nn\\
  & + \{f,x^\mu\}B_{\mu\nu}\{x_\nu,g\} + \{g,t_\mu\}C^{\mu\nu}\{t_\nu,f\} 
 \label{poissonbracket-Ansatz-coords}
\end{align}
where $B$ and $C$ are antisymmetric. 
This is not unique due to the constraints
\begin{align}
 0 &=  x^\mu\{t_\mu,.\} + t_\mu\{x^\mu,.\}  \nn\\
 0 &= R^{-2} x^\mu\{x_\mu,.\} + r^2 t_\mu\{t^\mu,.\} \ .
 \label{constraints-xt}
\end{align}
An explicit realization is given in  appendix \ref{sec:poisson-deriv-formulas}, which 
leads to the following formula 
\begin{align}
\boxed{\ 
 \{f,g\} = (\cD^\mu f) \{t_\mu,g\} + (\eth_\mu f) \{x_\mu,g\}
 \ } 
\label{poisson-bracket-realiz}
\end{align}
where $\cD$ and $\eth$ are derivations\footnote{Note that the present $\eth$ is
different from the one used in \cite{Sperling:2018xrm} for $H^4_N$.} on $\cC$ given by
\begin{align}
 \eth^\mu x^\r &=  \eta^{\r\nu}, 
  &\eth^\mu t^\r  &= - \frac{1}{R^2\cosh^{2}(\eta)}(t^\r x^\mu - x^\r t^\mu) \nn\\ 
 \cD^\mu t^\r &= P_\perp^{\r\mu} - r^2 \cosh^{-2}(\eta) t^\r t^\mu,    
 &\cD^\mu x^\r &=0 \ .
 \label{derivations-coords}
\end{align}
In particular, $\cD\cC^0=0$ means that its push-forward to $\cM$ vanishes, 
$\Pi_*\cD=0$. 
Thus $\cD$ are ``vertical`` derivatives along  the local $S^2$ fiber.
Moreover, we claim that  $\eth = \Pi_*\eth$ vanishes on the local $S^2$ fiber.
To see that, we  compute  
\begin{align}
 \eth^0 t^i  &= - \frac{1}{R\cosh(\eta)} t^i, \qquad \eth^j t^i = 0  \ 
\end{align}
at the reference point, and note that 
the first term reproduces the grow of 
$|t|^2 = r^{-2} \cosh^2(\eta) = -(rR)^{-2} |x|^2$ in time. 
Hence the push-forward of a Hamiltonian vector field to $\cM$ is given by
\begin{align}
\boxed{\ 
 \{g,.\}_\cM \equiv \Pi_*\{g,.\} = \{g,x^\mu\}\eth_\mu 
\ }
\end{align}
and similarly  the push-forward of the Poisson bracket to $\cM$ is given by
\begin{align}
\boxed{\ 
 \{f,g\}_\cM = \theta^{\mu\nu} \eth_\mu f \ \eth_\nu g
 \ }
  \label{reduced-bracket}
\end{align}
which is verified immediately for $f,g \in \cC^0$. 
%
For late times $\eta\gg 1$, we can replace $\eth$ by the simpler derivative
\begin{align}
 \del_\mu := \sinh^{-1}(\eta)\{t_\mu,.\}
 \label{del-def-basic}
\end{align}
since 
\begin{align}
 \del_\mu x^\nu &= \d^\nu_\mu, \nn\\
  \del_\mu t^\nu &= -\frac{1}{r^2 R^2}\sinh^{-1}(\eta)\theta^{\mu\nu}
   \sim -\frac{1}{R^2}\frac{1}{\cosh^2(\eta)} 
    (x^\mu t^\nu - x^\nu t^\mu)  
\end{align}
using \eq{theta-approx}, which 
agrees with $\eth$ at late times. We can then replace \eq{reduced-bracket} by 
\begin{align}
 \{f,g\}_\cM &\sim (\del_\mu f) \{x_\mu,g\} \ .
 \label{poisson-approx-x-1}
\end{align}
This will be used throughout this paper, and we discuss again its range of validity:

\paragraph{Asymptotic regime.}

The fact that $\cD$ vanishes on $\cC^0$ is very helpful to organize the
higher-spin theory: it leads to the estimate
\begin{align}
 |\cD f^{(s)}| \leq \frac s{|t|}|f| = \frac{r s}{\cosh(\eta)} |f|
\end{align}
for $f=f^{(s)}\in\cC^s$, where $|f|$ denotes the maximal value of $f$ on the internal $S^2$ 
over the particular point on $\cM$.
Thus
\begin{align}
 (\cD f)\{t_\mu,g\} &=   r  O(f \del g)  \nn\\
 (\del_\mu f) \{x_\mu,g\}  &=  r O(x\cdot\del f\del g)
\end{align} 
 for small spin and $\eta\gg 1$, using the abbreviation $x\cdot\del \equiv R\cosh(\eta)\del$. Thus
\begin{align}
 \{g,f\} &=  \{g,t_\mu\} (\cD^\mu f) +  \{g,x_\mu\}(\del_\mu f)  \nn\\
 %
 &=  \{g,x_\mu\}(\del_\mu f)\big(1 + O(\frac{f}{x\cdot\del f})\big)
  \label{poisson-approx-x-explicit}
\end{align}
so that we can use the 
approximation 
\begin{align}
\boxed{\ 
 \{g,f\} \sim  \{g,x^\mu\} (\del_\mu f)
 \ }
 \label{poisson-approx-x}
\end{align}
in the asymptotic regime, i.e. $x\cdot\del f \gg f$. 
If both $f$ and $g$ are in the asymptotic regime \eq{asymptotic-scale}, we can write 
\begin{align}
\boxed{\
 \{f,g\} \sim \theta^{\mu\nu} \del_\mu g \del_\nu f \sim \{f,g\}_\cM
 \ }
 \label{generic-commutator-estimate-2}
\end{align}
in agreement with \eq{generic-commutator-estimate}.
Of course we are free to use any coordinates in $\cC^0$ here.
Thus we have replaced the Poisson structure  $\C P^{1,2}$ by a 
$\hs$ valued bi-vector field on $\cM$,
which is its push-forward  
by the bundle projection as discussed in section \ref{sec:bundle-hs}.
This reduced  Poisson tensor will satisfy the following reduced
Jacobi identity 
\begin{align}
 \{\{f,g\}_\cM,h\}_\cM  +   \{\{g,h\}_\cM,f \}_\cM   + \{\{h,f\}_\cM,g \}_\cM \  \sim 0
\end{align}
in the asymptotic regime,
as long as $x\cdot \del \gg 1$.
This can also be seen using the exact identity 
\begin{align}
 0 = \{\{x^\mu,x^\nu\},x^\s\} + \{\{x^\nu,x^\s\},x^\mu \}  
    + \{\{x^\s,x^\mu\},x^\nu \} \
\end{align}
upon neglecting derivatives in the internal directions.
This in turn implies (cf. Appendix A in \cite{Steinacker:2010rh})
\begin{align}
 \del_\mu (\r_M \theta^{\mu\nu}) \sim 0, \qquad \r_M = \sqrt{|\theta^{\mu\nu}|}^{-1}
 \label{Poisson-conservation}
\end{align}
with the same qualifications. In Cartesian coordinates, 
this is simply the statement that $\theta^{\mu\nu} \sim const$
in the asymptotic regime. It would be desirable to find an exact 4-dimensional form of these 
relations.

\section{Higher-spin  gauge invariance and Lie derivatives}
\label{sec:gaugeinv-Lie}

\subsection{Gauge transformation of scalar fields}

Consider some scalar field  $\phi \in \cC$ on $\C P^{1,2}$, which transforms under
gauge transformations  as\footnote{This  
arises from the matrix model gauge transformation $\phi \to U \phi U^{-1}$ for $\phi \in \End(\cH)$. 
For $U = e^{i\L}$, its
infinitesimal version is $\phi \to -i[\L,\phi] \sim \{\L,\phi\}$ for $\L = \L^\dagger \in \End(\cH)$.}
\begin{align}
 \d_\L\phi &= \{\L, \phi \} = \cL_{\xi} \phi, 
       \qquad \L = \L^* \in \cC^\infty(\C P^{1,2}) \ .
 \label{gaugetrafo-1}
\end{align}
This is nothing but  the Lie derivative of $\phi \in \cC^\infty(\C P^{1,2})$ along
the Hamiltonian vector field $\xi = \{\L,.\}$.
For functions on $\cM\ $ i.e. restricted to $\phi \in \cC^0$, this reduces to
\begin{align}
 \d_\L \phi &= \cL_{\xi} \phi = \xi^\mu \del_\mu\phi
 \label{gauge-vectorfield-2}
\end{align}
which is interpreted as Lie derivative along the $\hs$-valued vector field 
\begin{align}
 \xi &= \xi^\mu \del_\mu, \qquad  \xi^\mu = \{\L,x^\mu\}\ .
 \label{diffeo-Lambda}
\end{align}
on $\cM$.
E.g. for $\L\in\cC^1$, its components are
\begin{align}
  \xi^\mu = \xi^\mu_{(0)} + \xi^\mu_{(2)} \quad \in \cC^0 \oplus \cC^2
\end{align}
and we can view  $\xi^\mu_{(0)} = \{\L,x^\mu\}_0 \in\cC^0$ as  vector field on  $\cM$ which
 defines a diffeomorphism 
 \begin{align}
 [\d_\L \phi]_0 &= \cL_{\xi_{(0)}} \phi .
\end{align}
However one typically cannot get rid of the component 
$\xi_{(2)} = \{\L,x^\mu\}_2  \in \cC^2$, and   
$\xi$ should be considered as generator of generalized $\hs$-valued diffeomorphisms.

\paragraph{$\hs$-valued scalar fields on $\cM$.}

In the asymptotic regime \eq{asymptotic-scale},  the above formula generalizes to arbitrary $\phi\in\cC$ as follows:
\begin{align}
 \d_\L\phi &= \{\L,\phi\} \sim  \{\L,x^\mu\} \del_\mu \phi  \nn\\
  & =  \xi^\mu  \del_\mu \phi  =: \bar\cL_\xi \phi \ .
  \label{scalar-Lie-full}
\end{align}
This is interpreted as Lie derivative of a $\hs$-valued function $\phi$
along the $\hs$-valued vector field $\xi$, indicated by $\bar\cL$. This makes sense as long as 
 $\phi$ and $\xi$ are in the asymptotic regime.

\subsection{Gauge transformation of vector fields}

Now consider some given Hamiltonian vector field $E=\{Z,.\}$  on $\C P^{1,2}$, such as the frame discussed below.
Assume that $Z$ transforms under a gauge transformation \eq{gaugetrafo-1} as 
\begin{align}
 \d_\L Z &= \{\L, Z\} \ .
\end{align}
Then the associated vector field $E$ transforms as
\begin{align}
 (\d_\L E) \phi  &= \{\{\L,Z\},\phi \}
  =  \{\L,\{ Z,\phi \}\} - \{ Z,\{\L,\phi \}\} \nn\\
   &= (\cL_\xi E)\phi 
  \label{Liederivative-id-1}
\end{align}
for any $\phi\in\cC$, where $\xi = \{\L,.\}$. 
This is  precisely the Lie derivative of the vector field $E$ along $\xi$, hence
\begin{align}
 \d_\L E = \cL_\xi E \ .
 \label{vectorfield-Lie}
\end{align}
This will give the transformation of the frame in \eq{frame-CP}, and it
extends to  higher-rank tensor fields as 
\begin{align}
 \d_\L (E \otimes E') = (\cL_\xi E)\otimes E' + E \otimes \cL_\xi E'
  = \cL_\xi(E \otimes E') \ 
  \label{tensor-trafo}
\end{align}
which will apply to the metric tensor.
We also recall that the Lie derivative of the Poisson bi-vector field along any Hamiltonian vector field
$\xi$ on $\C P^{1,2}$ vanishes $\cL_\xi\{.,.\} = 0$ , due to the Jacobi identity.

\paragraph{Gauge transformations and $\hs$-valued Lie derivative  on $\cM$.}

Now consider the same  from the 4-dimensional point of view.
As discussed in section \ref{sec:bundle-hs}, the push-forward of the Hamiltonian  vector field 
$E=\{Z,.\}$ to $\cM$ for $Z\in\cC$
defines a $\hs$-valued vector field  $E^\mu \del_\mu$ on  $\cM$, with components
\begin{align}
 E^\mu := \{Z,x^\mu\}  \ .
 \label{rank-1-field-def}
\end{align}
Under gauge transformations $\d_\L Z = \{\L,Z\}$, we can still use the first line in \eq{Liederivative-id-1},
\begin{align}
 (\d_\L E) \phi  &= \{\L,\{Z,\phi \}\} - \{Z,\{\L,\phi \}\} \ .
 \end{align} 
At this point it is more transparent to use coordinates $x^\mu$ on $\cM$, and the above becomes
\begin{align}
 \d_\L E^\mu &=   \{\L,E^\mu\} - \{Z,\xi^\mu\} \nn\\
 &\sim  \xi^\r \del_\r E^\mu  -  \{Z,x^\r\} \del_\r \xi^\mu \nn\\
 &=  \xi^\r \del_\r E^\mu -  E^\r \del_\r \xi^\mu \nn\\
 &=: \bar\cL_\xi E^\mu 
 \label{gauge-VF-M}
\end{align}
(for $\phi = x^\mu$), using the approximation \eq{poisson-approx-x} in the second line in the asymptotic regime.
This has the standard form of a Lie derivative along a field 
$\xi^\mu = \{\L,x^\mu\}$, and  constitutes our ``working definition'' of a $\hs$-valued Lie derivative $\bar\cL$ on $\cM$, 
generalizing \eq{scalar-Lie-full}.
Remember that  both  $E^\mu$
and $\xi^\mu$ have $\hs$ components.

%

The above considerations provide a 4-dimensional geometrical interpretation of 
the emergent diffeomorphism invariance and its higher spin extension, 
which arise\footnote{This generalizes analogous results in \cite{Langmann:2001yr} in a dimensionally reduced setting.} from
Hamiltonian vector fields on $\C P^{1,2}$.
Since these preserve the symplectic volume on $\C P^{1,2}$, the resulting $\hs$ diffeomorphisms on $\cM$
should be volume-preserving in some sense
also from the 4-dimensional point of view. 
This is indeed the case, as elaborated in Appendix \ref{sec:vol-pres-diffeo}.

\section{Kinetic action, frame and effective metric on $\cM$}
\label{sec:action-frame}

Now we want to understand the effective metric in the matrix model in the non-linear regime.
To keep the discussion simple, we focus on  the kinetic term 
for a scalar field $\phi$ in the matrix model,
\begin{align}
  S[\phi] &=  \Tr [Z^{{\dot\a}},\phi][Z_{{\dot\a}},\phi] = - \Tr \phi \Box \phi, 
  \qquad \dot\a = 0,...,3
  \label{scalar-action}
\end{align}
on a background given by some solution $Z_{{\dot\a}}$ of \eq{matrix-eom}.
Such scalar fields arise as transversal
brane fluctuation $\phi \equiv Z_b, \ b=4,...,9$.
The effective metric is encoded in the d'Alembertian $\Box$, and will therefore govern 
all fluctuations in the model.
We  study the above action in the semi-classical limit from
two different points of view.

\paragraph{$\C P^{1,2}$ point of view.}

Consider  a function $\phi$  on $\C P^{1,2}$. The background  $Z_{{\dot\a}}$ defines a  frame
\begin{align}
 E_{{\dot\a}} [\phi] &:= \{Z_{{\dot\a}},\phi\}  , \qquad \dot\a = 0,...,3
 \label{frame-CP}
\end{align} 
which are derivations on $\cC = \cC^\infty(\C P^{1,2})$.
Frame indices will always be dotted Greek letters ${\dot\a}, \dot\b,...$,
which transform under the global $SO(3,1)$.
This defines a metric as a bi-vector field 
\begin{align}
 \g[\phi,\psi] &:= \eta^{{{\dot\a}}{\dot\b}} E_{{\dot\a}}[\phi] E_{\dot\b}[\psi]
\end{align}
and the action \eq{scalar-action} can be written in the semi-classical limit as\footnote{The 6-dimensional action 
could be rewritten in 
covariant form using an effective metric along the lines of \cite{Steinacker:2010rh}, but
we refrain from doing so because we want to emphasize the 4-dimensional
point of view.} 
\begin{align}
  S[\phi] &=  \Tr [Z^{{\dot\a}},\phi][Z_{{\dot\a}},\phi] 
  \ \sim \ - \int_{\C P^{1,2}} \omega^{\wedge 3}\,\{Z^{{\dot\a}},\phi\}\{Z_{{\dot\a}},\phi\}  \nn\\
   &= \ -\int_{\C P^{1,2}} \dd^6 z\,\frac 1{\sqrt{|\theta^{AB}|}}  \, \g^{AB}\del_A \phi \del_B \phi \ 
 \end{align}
 where e.g. $y^A = (x^\mu,\vartheta,\varphi)$ are coordinates on $\C P^{1,2}$.
 Here the coordinate form of the frame is 
\begin{align}
 E_{{\dot\a}}^A &:= \{Z_{{\dot\a}},y^A\}, 
 \qquad \g^{AB} = \eta^{{{\dot\a}}{\dot\b}} E^A_{{\dot\a}} E^B_{\dot\b}
\end{align}
such that
$E_{{\dot\a}}[\phi] = E_{{\dot\a}}^A \del_A \phi\ $ for $\phi\in\cC$,
 and the $SO(4,2)$-invariant symplectic volume form on $\C P^{1,2}$ is
\begin{align}
 \omega^{\wedge 3} &= \Omega_{A_1 ... A_6} dy^{A_1} ... dy^{A_6} 
  = \frac 1{\sqrt{|\theta^{AB}|}} dy^{1} ... dy^{6} \ .
 \label{6-volume-form}
\end{align} 
However, this does not properly reflect the structure of $\C P^{1,2}$ as $S^2$ bundle over $\cM$,
and we will rewrite it in a way which is more transparent from the 4-dimensional point of view.

\paragraph{3+1-dimensional point of view on $\cM$.}

Now  consider the reduction (more precisely the push-forward) of the same configuration to $\cM$.
This applies automatically if $\phi$ is in the asymptotic regime.
Then 
 \begin{align}
  E_{{\dot\a}}[\phi] \sim \tensor{E}{_{\dot\a}^\mu} \del_\mu \phi, \qquad \tensor{E}{_{\dot\a}^\mu} 
     &:= \{Z_{{\dot\a}},x^\mu\}
  \label{frame-M}
 \end{align}
 is a $\hs$-valued frame on $\cM$, which is an invertible $\cC$-valued $4\times 4$ matrix.
 Instead of the Cartesian coordinates $x^\mu$ we could use any other coordinates on $\cM$.
 We can write the volume form \eq{6-volume-form}  as 
 \begin{align}
 \omega^{\wedge 3} &= \r_M dx^{0} ... dx^{3} \Omega_2, \qquad \r_M = \sinh^{-1}(\eta)
 \label{4-volume-form}
\end{align}
 where $\Omega_2$ is the volume form of the unit 2-sphere corresponding to the local fiber,
 and $\r_M$ is the corresponding 4-density.
 The explicit form is easily obtained as the unique $SO(4,1)$-invariant volume on $H^4$, 
 cf. \cite{Steinacker:2017bhb}.
Then the  action reduces to  
  \begin{align}
  S[\phi] &\sim \ -\int _{\cM^{3,1}} \!\!  d x_0 \ldots d x_3 \, \rho_M \,  \g^{\mu\nu}
   \del_\mu \phi \del_\nu \phi \ 
   = \ -\int_{\cM^{3,1}}\!\! \dd^4 x\, \sqrt{|G_{\mu\nu}|}\,
  G^{\mu\nu}\del_\mu \varphi \del_\nu \varphi \ 
  \label{scalar-action-G}
\end{align} 
absorbing
some dimensionful constants  in $\varphi \sim \phi$. Here
\begin{align}
 \g^{\mu\nu} &= \eta^{{{\dot\a}}{\dot\b}} \tensor{E}{_{\dot\a}^\mu} \tensor{E}{_{\dot\b}^\nu}, \qquad 
 G^{\mu\nu} := \frac 1{\r^2} \g^{\mu\nu}
 \label{eff-metric}
\end{align}
are an auxiliary  and the effective ($\hs$-valued) metric on $\cM$, respectively, 
and
\begin{align}
\r^2 &= \r_M \sqrt{|\g^{\mu\nu}|}  \ 
 \label{rho-def}
\end{align}
is a scalar field (rather than a density).
One can also introduce a  densitized vielbein  
which directly gives the effective metric,
\begin{align}
 \tensor{\cE}{_{\dot\a}^\mu} = \r^{-1} \tensor{E}{_{\dot\a}^\mu},
 \qquad G^{\mu\nu} &= \eta^{{{\dot\a}}{\dot\b}} \tensor{\cE}{_{\dot\a}^\mu} \tensor{\cE}{_{\dot\b}^\nu} \ .
 \label{rescaled-frame-def}
\end{align}
The last form of \eq{scalar-action-G} is indeed $\hs$ covariant,
in the sense discussed  below.
This is  the dominant contribution for fields 
$\phi\in\cC^s$ in the asymptotic regime.
The metric $G^{\mu\nu}$ indeed governs the d'Alembertian \eq{Box-def},
\begin{align}
 \Box &= - \{Z^{{\dot\a}},\{Z_{{\dot\a}},.\}\} = \r^2 \Box_{G}
\end{align}
as shown in Lemma \eq{lem:Box}.
At the linearized level and for scalar fields $\phi\in\cC^0$, only the 
scalar component 
$[G^{\mu\nu}]_0  \ \in\cC^0$ contributes.

For the background solution $\bar Z_{{\dot\a}} = T_{{\dot\a}} = R^{-1} \cM_{{\dot\a}4}$ \eq{T-solution}, 
we obtain in  Cartesian coordinates 
\begin{align}
  \bar E_{{\dot\a}}^\mu &= \sinh(\eta)\d_{{\dot\a}}^\mu  \ , &  \bar\g^{\mu\nu} &= \sinh^2(\eta) \eta^{\mu\nu}, \nn\\
  \bar G^{\mu\nu} &=  \sinh^{-3}(\eta)\bar\g^{\mu\nu} = \sinh^{-1}(\eta) \eta^{\mu\nu} , &
  \bar\r^{2} &= \sinh^3(\eta) \ .
   \label{eff-metric-G-bar}
\end{align}

\paragraph{Gauge transformation of the frame.}

Consider first the  $\C P^{1,2}$ point of view.
Under a gauge transformation \eq{gaugetrafo-1}, the background transforms  as 
\begin{align}
 \d_\L Z^{{\dot\a}} &= \{\L, Z^{{\dot\a}} \}
\end{align}
and the associated vector fields $E_{{\dot\a}}$ on $\C P^{1,2}$  transform as
\begin{align}
 \d_\L E_{{\dot\a}} = \cL_\xi E_{{\dot\a}} \ 
\end{align}
using \eq{vectorfield-Lie},  where $\xi = \{\L,.\}$. 
In local coordinates $y^A$, this gives the coordinate expression
\begin{align}
 \d_\L E_{{\dot\a}}^A 
   = \xi^B\del_B E_{{\dot\a}}^A - E_{{\dot\a}}^B \del_B \xi^A \
  = \cL_\xi  E_{{\dot\a}}^A
\end{align}
and we  obtain
\begin{align}
 \d_\L \g^{AB} &= \d_\L(E_{{\dot\a}}^A E_{\dot\b}^B \eta^{{{\dot\a}}{\dot\b}}) = \cL_\xi \g^{AB} \ .
\end{align}

\paragraph{Reduction to $\cM$.}

Now consider the gauge transformations of the frame for $\hs$-valued functions on $\cM$,
in the asymptotic regime where all fields 
are in the asymptotic regime \eq{asymptotic-scale}.
Then we can use \eq{gauge-VF-M}, 
\begin{align}
 \d_\L E_{{\dot\a}}^\mu 
  \sim \bar\cL_\xi \tensor{E}{_{\dot\a}^\mu}
  \label{gaugetrafo-frame}
\end{align}
where $\xi^\nu = \{\L,x^\nu\}$ is a $\hs$-valued vector field,  
and $E_{{\dot\a}}^\mu$ is the $\hs$-valued frame \eq{frame-M} on $\cM$.
In particular, this implies 
\begin{align}
 \d_\L \g^{\mu\nu} \sim \bar\cL_\xi \g^{\mu\nu}
 = \nabla_{(\g)}^\mu\xi^\nu + \nabla_{(\g)}^\nu \xi^\mu \ 
\end{align}
where $\nabla_{(\g)}$ is  the Levi-Civita connection corresponding to $\g^{\mu\nu}$.
We would like to generalize this to the effective metric $G^{\mu\nu}$ 
including the conformal factor. To see this,
observe that the invariance of the symplectic form  $\cL_\xi\omega = 0$ implies 
\begin{align}
 \d_\L \r^4 \sim \bar\cL_\xi \r^4
 \label{rho-invar-M}
\end{align}
for 
\begin{align}
 \r^4 &:= \g^{\mu_0\nu_0}...\g^{\mu_3 \nu_3}  \Omega_{\mu_0 ... \mu_3}   \Omega_{\nu_0 ... \nu_3}
  = |\g^{\mu\nu}| |\omega_{\mu\nu}| = |\g^{\mu\nu}|\r_M^2 \ ,
\end{align}
where 
\begin{align}
\Omega_{\mu_0 ... \mu_3} &:= (\omega^{\wedge 2})_{\mu_0 ... \mu_3}, 
\qquad 
 \Omega_{\mu_0...\mu_3}\varepsilon^{\mu_0..\mu_3} =  \sqrt{|\omega_{\mu\nu}|} = \r_M
\end{align}
is the effective volume form \eq{4-volume-form} on $\cM$.
Thus $\r^2$ coincides with the conformal factor for the effective metric \eq{eff-metric},
and we obtain 
\begin{align}
 \d_\L G^{\mu\nu} &\sim \bar\cL_\xi G^{\mu\nu} \sim \nabla_{(G)}^\mu\xi^\nu + \nabla_{(G)}^\nu \xi^\mu
 \label{effective-G-gaugetrafo}
\end{align}
for the effective metric in the asymptotic regime. Here
$\nabla_{(G)}$ is  the Levi-Civita connection corresponding to $G^{\mu\nu}$.

To summarize, the  gauge invariance arising from symplectomorphisms on $\C P^{1,2}$
leads to an emergent higher-spin symmetry  in 4 dimensions. This is achieved by considering 
$\tensor{E}{_{\dot\a}^\mu}, \ G^{\mu\nu}$  and  $\xi^\mu$ as $\hs$-valued vector fields on $\cM^{3,1}$, and 
\eq{effective-G-gaugetrafo} should be understood in this higher spin sense.
In the linearized regime for  $s=1$, this reduces to 
the standard formulas for volume-preserving diffeomorphisms in 4 dimensions
\cite{Steinacker:2019dii}.

\section{Geometric description of the non-linear regime}

In this section, we develop a geometric formalism based on a higher-spin generalization of the Weitzenb\"ock 
connection and torsion. We will work in the asymptotic regime \eq{asymptotic-scale}, using the 4-dimensional 
 point of view developed above.

\subsection{Weitzenb\"ock connection and torsion}

The fundamental degrees of freedom of the matrix model is the 
background $Z_{{\dot\a}}$ and its associated 
vielbein $E_{{\dot\a}} = \{Z_{{\dot\a}},.\}$. 
It is then natural to 
define a connection which respects the vielbein,
\begin{align}
 0 = \nabla_{\dot\g} \tensor{E}{_{\dot\a}^\mu} &= E_{\dot\g}[\tensor{E}{_{\dot\a}^\mu}]  
 + \tensor{\Gamma}{_{\dot \g}_\r^\mu} \tensor{E}{_{\dot\a}^\r} 
\end{align}
analogous to the Weitzenb\"ock connection \cite{aldrovandi2012teleparallel}, along $E_{\dot\g}$.
This can always be solved as
\begin{align}
\boxed{
 \tensor{\Gamma}{_{\dot\g}_{\dot\a}^\mu} := - E_{\dot\g}[\tensor{E}{_{\dot\a}^\mu}] \ 
 = \tensor{\Gamma}{_{\dot \g}_\r^\mu}\tensor{E}{_{\dot\a}^\r} \quad \in \cC \
}
\label{Weitzenbock-explicit}
\end{align}
provided $E^{\r}_{{\dot\a}}$ is an invertible matrix taking values in $\cC$.
In the perturbative regime, the most significant contribution should be the $\cC^0$ components of 
$\tensor{E}{_{\dot\a}^\mu}$ and $\tensor{\Gamma}{_{\dot \g}_\r^\mu}$, accompanied by 
some  higher-spin contributions.
In local coordinates $y^\mu$, this is
\begin{align}
 \del_\nu\tensor{E}{_{\dot\a}^\mu} \ 
 &= -\tensor{\Gamma}{_{\nu}_\r^\mu}\tensor{E}{_{\dot\a}^\r} \ 
  \label{weizenbock-gamma-expl}
\end{align}
where we define the derivation\footnote{for the background $Z_{\dot\a} = t_{\dot\a}$,
the $\del_\nu$ agrees with $\eth_\nu$ and \eq{del-def-basic} on $\cC^0$, and 
possible differences on $\hs$ are negligible in the asymptotic regime.}
\begin{align}
 \del_\nu &:= \tensor{E}{^{\dot\a}_\nu}\{Z_{{\dot\a}},.\}, \qquad
 \del_\nu y^\mu = \d^\mu_\nu \ .
\end{align}
The inverse vierbein is defined as usual
\begin{align}
   \tensor{E}{^{\dot\a}_\mu}\tensor{E}{_{\dot\b}^\mu} &= \d^{\dot\a}_{\dot\b}, &  
    \g^{\mu\nu} &= \eta^{{{\dot\a}}{\dot\b}} \tensor{E}{_{\dot\a}^\mu} \tensor{E}{_{\dot\b}^\nu},  \nn\\
    \tensor{E}{^{\dot\a}_\nu} \tensor{E}{_{\dot\a}^\mu} &= \d^\mu_\nu,  & 
  \eta_{{{\dot\a}}{\dot\b}} &= \tensor{E}{_{\dot\a}^\mu} \tensor{E}{_{\dot\b}^\nu}  \g_{\mu\nu} \ .
\end{align}
This connection is automatically compatible
with the metric $\g^{\mu\nu}$, 
\begin{align}
 \nabla \g^{\mu\nu} = \eta^{{{\dot\a}}{\dot\b}} (\nabla \tensor{E}{_{\dot\a}^\mu} \tensor{E}{_{\dot\b}^\nu}
      +  \tensor{E}{_{\dot\a}^\mu}\nabla \tensor{E}{_{\dot\b}^\nu}) = 0 \ .
\end{align}
For any $\hs$-valued vector field $V^\mu$ on $\cM$, we can then define
the  covariant derivative as
\begin{align}
 \nabla_\mu V^\nu &=  \del_\mu V^\nu + \tensor{\Gamma}{_\mu_\r^\nu} V^\r  \ .
\end{align}
This connection is flat\footnote{The curvature on $\cM$ is defined  
as usual by
$\cR_{{\dot\a},\dot\b}[E_{\dot\g}] :=[\nabla_{{\dot\a}},\nabla_{\dot\b}]E_{\dot\g} - \nabla_{[E_{{\dot\a}},E_{\dot\b}]} E_{\dot\g} = 0$
{\em in the asymptotic regime}, where $[E_{{\dot\a}},E_{\dot\b}]$ 
is a linear combination of the $E_{\dot\g}$. However, there is no fully noncommutative version.} 
since the frame is parallel, $\nabla E_{\dot\b} = 0$.
However it typically has torsion,
\begin{align}
 T[X,Y] = \nabla_X Y - \nabla_Y X - [X,Y]
\end{align}
which for the frame  can be computed as
\begin{align}
 T_{{\dot\g}{\dot\b}} &\equiv T[ E_{\dot\g},E_{\dot\b}] = \nabla_{\dot\g} E_{\dot\b} - \nabla_{\dot\b}  E_{\dot\g} 
   - [E_{\dot\g}, E_{\dot\b}] \nn\\
  &= - [E_{\dot\g},E_{\dot\b}] \equiv -\{Z_{\dot\g},\{Z_{\dot\b},.\} \} + \{Z_{\dot\b},\{Z_{\dot\g},.\} \}  \nn\\
  &= \{\hat\Theta_{{\dot\g}{\dot\b}},.\} \ , \nn\\
  \tensor{T}{_{\dot\a}_{\dot\b}^\mu} &:= \{\hat\Theta_{{{\dot\a}}{\dot\b}},y^\mu\} , 
  \qquad \hat\Theta_{{{\dot\a}}{\dot\b}} :=  -\{Z_{{\dot\a}},Z_{{\dot\b}}\}  
 \label{torsion-explicit}
\end{align}
using the Jacobi identity. 
Due to Lemma \ref{lemma:uniqueness-VF-C0}, $\tensor{T}{_{\dot\a}_{\dot\b}^\mu}$ fully captures 
the noncommutative field strength
$\hat\Theta_{{{\dot\a}}{\dot\b}}$.
Thus torsion encodes the quantum structure of space-time, and it is the semi-classical
shadow of it. 
This is the key to understand the matrix model in terms of gravity\footnote{The role 
of torsion in a dimensionally reduced noncommutative gauge theory as 
related to gravity was already pointed out in \cite{Langmann:2001yr},
however the specifics are different. 
Torsion arises in \eq{torsion-explicit} as {\em derivative} of the NC field strength, unlike in \cite{Langmann:2001yr}.
Also, the action in previous work 
is typically given by a contraction of torsion, which is not the case here.
See also e.g. \cite{Ciric:2016isg} for other work related to torsion in a similar context.},
and the  equations of motion of the model 
will be re-formulated in terms of  torsion below.
For a perturbed cosmic background $Z_{{\dot\a}} = t_{{\dot\a}} + \cA_{{\dot\a}}$, we have
\begin{align}
 \hat\Theta_{{{\dot\a}}{\dot\b}} &= \frac{1}{r^2 R^2}\theta_{{{\dot\a}}{\dot\b}} 
 -(\{t_{{\dot\a}},\cA_{\dot\b}\} - \{t_{\dot\b},\cA_{{\dot\a}}\} + \{\cA_{{\dot\a}},\cA_{\dot\b}\}) \  ,
\end{align}
so that that  $\hat\Theta_{{{\dot\a}}{\dot\b}}\in \cC^1$ up to higher-spin corrections.
More explicitly, the torsion tensor is 
\begin{align}
 \tensor{T}{_{\dot\a}_{\dot\b}^\mu} &= \nabla_{{\dot\a}} \tensor{E}{_{\dot\b}^\mu} - \nabla_{\dot\b} \tensor{E}{_{\dot\a}^\mu} 
          - [E_{{\dot\a}},E_{\dot\b}]^\mu    \nn\\
    &= E_{{\dot\a}}[\tensor{E}{_{\dot\b}^\mu}] - E_{\dot\b}[ \tensor{E}{_{\dot\a}^\mu}] 
       + \tensor{\G}{_{\dot\a}_\r^\mu}\tensor{E}{_{\dot\b}^\r} - \tensor{\G}{_{\dot\b}_\r^\mu}\tensor{E}{_{\dot\a}^\r} 
       - [E_{{\dot\a}},E_{\dot\b}]^\mu  \nn\\
    &= \tensor{\G}{_{\dot\a}_\r^\mu}\tensor{E}{_{\dot\b}^\r} - \tensor{\G}{_{\dot\b}_\r^\mu}\tensor{E}{_{\dot\a}^\r}   \
    = \tensor{\G}{_{\dot\a}_{\dot\b}^\mu} - \tensor{\G}{_{\dot\b}_{\dot\a}^\mu} \nn \\
 \tensor{T}{_{\mu}_{\nu}^\r} &=\tensor{\G}{_{\mu}_{\nu}^\r} - \tensor{\G}{_{\nu}_{\mu}^\r} \ 
\end{align}
using the Jacobi identity
$E_{{\dot\a}}[E^\mu_{\dot\b}] - E_{\dot\b}[E^\mu_{{\dot\a}}]  \equiv [E_{{\dot\a}},E_{\dot\b}]^\mu$.
The  gauge transformation of the torsion tensor in the asymptotic regime is obtained again from \eq{gauge-VF-M},
\begin{align}
 \d_\L \hat\Theta_{{{\dot\a}}{\dot\b}} &= \{\L, \hat\Theta_{{{\dot\a}}{\dot\b}}\} \ , \nn\\
 \d_\L \tensor{T}{_{\dot\a}_{\dot\b}^\mu}
   &\sim \xi^\nu\del_\nu \tensor{T}{_{\dot\a}_{\dot\b}^\mu} - \tensor{T}{_{\dot\a}_{\dot\b}^\nu} \del_\nu \xi^\mu \
  = \bar\cL_\xi \tensor{T}{_{\dot\a}_{\dot\b}^\mu}
\end{align}
where $\xi^\nu = \{\L,x^\nu\}$ is a $\hs$-valued vector field on $\cM$. 
Together with \eq{gaugetrafo-frame} this implies 
\begin{align}
 \d_\L \tensor{T}{_{\r}_{\s}^\mu} \sim \bar\cL_\xi \tensor{T}{_{\r}_{\s}^\mu}
\end{align}
and similarly for the effective frame using \eq{rho-invar-M}.
Hence torsion transforms as a covariant tensor, just like the metric.

\paragraph{Relation with the effective Levi-Civita connection.}

Now consider 
the Levi-Civita connection $\nabla^{(\g)}$  for the  metric $\g^{\mu\nu}$, which is obtained 
as usual from the Christoffel symbols
\begin{align}
  \tensor{\G}{^{(\g)}_\mu_\nu^\s}
  &= \frac 12 \g^{\s\r}\Big(\del_\mu \g_{\r\nu} 
   + \del_\nu \g_{\r\mu}
  - \del_\r \g_{\mu\nu}\Big) \nn\\
  &=   \frac 12 \g^{\s\r}\Big(\tensor{\G}{_\mu_\r_\nu} + \tensor{\G}{_\mu_\nu_\r}
   + \tensor{\G}{_\nu_\r_\mu} + \tensor{\G}{_\nu_\mu_\r}
  - \tensor{\G}{_\r_\mu_\nu} - \tensor{\G}{_\r_\nu_\mu}\Big) \nn\\
  %
   &= \tensor{\G}{_\mu_\nu^\s}  - \tensor{K}{_\mu_\nu^\s} \ .
   \label{LC-contorsion-1}
\end{align}
Here
\begin{align}
 \tensor{K}{_{\mu}_{\nu}^{\s}}
 &= \frac 12 (\tensor{T}{_{\mu}_{\nu}^{\s}} 
             + \tensor{T}{^{\s}_{\mu}_{\nu}} 
             - \tensor{T}{_{\nu}^{\s}_{\mu}})
 = -  \tensor{K}{_\mu^\s_\nu} \quad \in \cC
  \label{Levi-contorsion-basic}
\end{align}
is (a higher-spin analog of) the contorsion of the basic Weitzenb\"ock connection, 
which is antisymmetric in ${\nu}{\s}$.
Therefore
\begin{align}
 \tensor{\Gamma}{_\mu_\nu^\r} &= \tensor{\G}{^{(\g)}_\mu_\nu^\r} + \tensor{K}{_\mu_\nu^\r} \nn\\
 \nabla_{{\mu}} V^\nu &= \nabla^{(\g)}_{{\mu}} V^\nu + \tensor{K}{_{\mu}_\r^\nu} V^\r \ .
 \label{relation-contorsion-levi}
\end{align}
Similarly, the Levi-Civita connection $\nabla^{(G)}$  for the effective metric $G^{\mu\nu}$ 
 is obtained as  
\begin{align}
  \tensor{\G}{^{(G)}_\mu_\nu^\s}
 &= \frac 12 G^{\s\r}\Big(\del_\mu G_{\r\nu} 
   + \del_\nu G_{\r\mu}
  - \del_\r G_{\mu\nu}\Big) \nn\\
 &= \frac 12\r^{-2}\Big( \d^\s_\nu\del_\mu \r^2 
   +  \d^\s_\mu\del_\nu \r^2 
  -  \g_{\mu\nu}\g^{\s\r}\del_\r \r^2 \Big) 
  +  \frac 12 \g^{\s\r}\Big(\del_\mu \g_{\r\nu} 
   + \del_\nu \g_{\r\mu}
  - \del_\r \g_{\mu\nu}\Big)  
  %
  %
\end{align}
which together with the above gives
\begin{align}
\boxed{\ 
  \tensor{\G}{^{(G)}_\mu_\nu^\s}
 = \tensor{\tilde\G}{_\mu_\nu^\s}  - \tensor{\cK}{_\mu_\nu^\s} 
 = \tensor{\G}{_\mu_\nu^\s} + \d^\s_\nu \r^{-1} \del_\mu \r - \tensor{\cK}{_\mu_\nu^\s} 
\ }
\label{LC-contorsion-eff}
\end{align}
Here
\begin{align}
 \tensor{\tilde\G}{_\mu_\nu^\s} 
  &:= \tensor{\G}{_\mu_\nu^\s} + \d^\s_\nu \r^{-1} \del_\mu \r \ , \nn\\
 \tensor{\cK}{_\mu_\nu^\s} &=
   \tensor{K}{_{\mu}_{\nu}^{\s}}
  + \Big(G_{\mu\nu}\r^{-1} \del^\s \r - \d^\s_{\mu}\r^{-1}\del_\nu \r\Big) 
   = -  \tensor{\cK}{_\mu^\s_\nu} 
  \label{Levi-contorsion-full}
\end{align}
 will be recognized below as Weitzenb\"ock connection and
contorsion of the effective frame, and accordingly the indices should be raised and lowered with $G^{\mu\nu}$.
To avoid any confusion with the two metrics $\g^{\mu\nu}$ and $G^{\mu\nu}$,
all connection and (con)torsion symbols will be written with two lower and 
one upper index, where no ambiguity arises.
%
This allows to rewrite the effective Levi-Civita connection in terms of the 
Weitzenb\"ock connection and the torsion,
which amounts to the simple rule for the covariant derivatives
\begin{align}
 \nabla^{(G)}_\mu V^{\s} &= \nabla_\mu V^{\s}
 - \tensor{\cK}{_\mu_\nu^\s} V^\nu
  +  \r^{-1}\del_{\mu}\r\, V^\s \nn\\
\nabla^{(G)}_\mu V_\s  &= \nabla_\mu V_\s
  + \tensor{\cK}{_\mu_\s^\nu} V_\nu
   - \r^{-1}\del_{\mu}\r\, V_\s
   \label{nabla-nablaLC-rel}
\end{align}
and similarly for higher-rank tensors.

\paragraph{Effective (rescaled) frame.}

The rescaled or effective frame \eq{rescaled-frame-def} for the effective metric\footnote{Note that 
$\tensor{\cE}{_{\dot\a}^\mu}$ is {\em not} the Hamiltonian vector field 
 associated to $\frac{1}{\rho} \cE_{\dot\a}$.}
\begin{align}
 \tensor{\cE}{_{\dot\a}^\mu} = \frac{1}{\rho} \tensor{E}{_{\dot\a}^\mu} 
 = \frac{1}{\rho} \{Z_{\dot\a},y^\mu\}
\end{align}
gives rise to an associated Weitzenb\"ock connection 
which is compatible with the effective metric,
\begin{align}
\tilde\nabla \tensor{\cE}{_{\dot\a}^\mu} = 0 =
\tilde\nabla G^{\mu\nu} \ .
\end{align}
 One must  be very careful with the frame indices, since there are two different frames in the game.
 It is therefore safer to use the coordinate form. 
 Then 
 \begin{align}
 \tensor{\tilde\Gamma}{_{\nu}_{\dot\a}^\mu} 
 &= - \del_\nu\tensor{\cE}{_{\dot\a}^\mu} \ 
 = -  \del_{\nu}\big(\r^{-1}\tensor{E}{_{\dot\a}^\mu}\big)
   =  \r^{-1}\tensor{\Gamma}{_{\nu}_{\dot\a}^\mu}
    + \r^{-1} \del_\nu\rho\tensor{\cE}{_{\dot\a}^\mu} \
     =: \tensor{\tilde\Gamma}{_{\nu}_\r^\mu}\tensor{\cE}{_{\dot\a}^\r} \ \nn\\
 \tensor{\tilde\Gamma}{_{\nu}_\s^\mu} &=  \tensor{\Gamma}{_{\nu}_{\s}^\mu}
   + \r^{-1} \d^\mu_\s \, \del_\nu\rho\ .
 \label{eff-Gamma-rel}    
 \end{align}  
 For the covariant derivatives, this amounts to the simple rule
\begin{align}
 \tilde\nabla_\mu V^{\s} &= \nabla_\mu V^{\s}
  +  \r^{-1}\del_{\mu}\r\, V^\s \nn\\
\tilde\nabla_\mu V_\s  &= \nabla_\mu V_\s
   - \r^{-1}\del_{\mu}\r\, V_\s
   \label{nabla-nablatilde-rel}
\end{align}
and similarly for higher-rank tensors.
Then the torsion tensor is
\begin{align}
 \tensor{\cT}{_{\mu}_{\nu}^\s} &= \tensor{\tilde \G}{_{\mu}_{\nu}^\s} - \tensor{\tilde \G}{_{\nu}_{\mu}^\s} \ 
  = \tensor{T}{_{\mu}_{\nu}^\s} + \r^{-1} \big(\d_{\nu}^\s\del_{\mu}\rho - \d_{\mu}^\s\del_{\nu}\rho \big) \ 
    \label{tilde-T-T}
\end{align}
and the effective contorsion is related to that of the basic frame as follows 
\begin{align}
 \tensor{\cK}{_{\mu}_{\nu}_\s} 
  &= \tensor{K}{_{\mu}_{\nu}_\s}  +\r^{-1} \big(G_{\mu\nu} \del_\s\r - G_{\mu\s} \del_\nu\r  \big)
   = -\tensor{\cK}{_{\mu}_{\s}_{\nu}} 
  \label{tilde-K-K}
\end{align}
in complete agreement with \eq{Levi-contorsion-full}.
Calligraphic fonts (or a tilde) indicate the rescaled frame.

\paragraph{(Harmonic) normal coordinates.}

We will denote  coordinates around some point $p\in\cM$
as {\em normal coordinates} at $p$ if the Christoffel symbols vanish at $p$,
\begin{align}
 {\tensor{\G}{^{(G)}_{\nu\s}^\mu}}|_p &= 0 \ .
\end{align}
Then the Weitzenb\"ock connection coincides with the contorsion due to \eq{LC-contorsion-eff},
\begin{align}
 \qquad \tensor{\cK}{_{\mu}_{\nu}^{\s}} = \tensor{\tilde\G}{_{\mu}_{\nu}^{\s}} \qquad \mbox{at} \ \ p \ .
\label{Weitzenbock-Torsion-RiemanNC}
\end{align}
Hence if torsion vanishes, the Levi-Civita connection coincides with the 
Weitzenb\"ock connection, and both are flat.
This can be done either for the basic metric $\g^{\mu\nu}$ or 
for  the effective metric $G^{\mu\nu}$, but typically not for both simultaneously.
In the present context, the Christoffel symbols
will be $\hs$-valued in general, and so are\footnote{Using $\hs$-valued coordinates $y^\mu$
on $\cM$ simply amounts to a deformation of the 
bundle projection $\Pi$ \eq{projection-hs}.
The algebra generated by these 4 generators $y^\mu$ is formally
isomorphic to the algebra of functions on  $\cC^\infty(\R^{3,1})$.} the normal coordinates $y^\mu$.

In particular, consider harmonic coordinates local coordinates $y^\mu$ on $\cM$ around $p$,
which by definition satisfy  $\Box_G  y^\mu = 0$.
We can demand {\em in addition} that at any given point they are also normal coordinates, so that
\begin{align}
 \Box_G  y^\mu &\equiv 0 \nn\\
  {\tensor{\G}{^{(G)}_{\nu\s}^\mu}}|_p &= 0
  \label{harm-NC-def}
\end{align}
To see that they exist, it suffices to note that one can find harmonic functions $y^\mu$
with any prescribed ``boundary value'' and normal derivative for any hypersurface
through $p$;
thus one can choose the ${\tensor{\G}{^{(G)}_{\nu\s}^\mu}}|_p$
freely up to $\G_{(G)}^\mu$, which vanishes by the harmonic condition.
We denote such coordinates as ``harmonic normal coordinates'', which will be used in the alternative 
derivation of the Ricci tensor in section \ref{sec:alt-Ricci}.

\subsection{Some useful identities}

It was shown in \cite{Steinacker:2010rh} that for symplectic manifolds,
the Matrix or Poisson d'Alembertian $\Box$ \eq{Box-def} is proportional to the  metric d'Alembertian 
\begin{align}
 \Box_G \phi = - \frac{1}{\sqrt{|G_{\mu\nu}|}}
    \del_\mu\big(\sqrt{|G_{\mu\nu}|}G^{\mu\nu}\del_\nu\phi\big)
\end{align}
The same relation holds  
in the present reduced 4-dimensional setting\footnote{This is bound to hold 
due to the covariant form \eq{scalar-action-G} of the kinetic term.}:

\begin{lem}
 \label{lem:Box}
 The Matrix or Poisson d'Alembertian $\Box$ \eq{Box-def} is related to the 
 metric d'Alembertian  $\Box_G \phi$ 
via
\begin{align}
 \Box\phi = -\{Z_{\dot\a},\{Z^{\dot\a},\phi\}\} 
 = \r^2 \Box_G \phi 
 \label{Box-Box-relation}
\end{align}
in the asymptotic regime\footnote{The only ``approximation'' used is \eq{Poisson-conservation}, which is established  
 in the asymptotic regime only, although it appears to hold more generally.},
 with $\rho$ given in \eq{rho-def}.
Furthermore, the following identities hold
\begin{align}
 \tensor{\G}{_\r_\mu^\r} &= \r_M^{-1}\del_\mu \r_M \
  \label{Gamma-mu-t-mu-id}  \\
  \Box y^\s &= \del_\mu\g^{\mu\s}  
    + \tensor{\G}{_\r^\s^\r}  
   = -\tensor{\G}{_\mu^\mu^\s} 
   = \tensor{\Gamma}{_{\nu}_\r^\mu} \g^{\nu\r} 
   = \tensor{\Gamma}{^{(G)}_{\nu}_\r^\mu} \g^{\nu\r} \label{Gamma-contract-id}\\
\tensor{\G}{_{\r}_\mu^\mu} &= - \del_\r\ln(\sqrt{|\gamma^{\mu\nu}|}) \ .
   \label{Gamma-det-identity}
\end{align}

\end{lem}

This is proved in appendix \ref{sec:proof-Lemma-Box},  by writing out  $\Box$ using vielbein and Weitzenb\"ock 
connection and using \eq{Poisson-conservation}. 
This is a generalization of a similar result in \cite{Steinacker:2010rh} for symplectic branes.
It leads to the following explicit formulas for the contraction of the 
(con)torsion:

\begin{lem}
\label{lem:contraction-torsion}
The (con)torsion associated to the basic frame  satisfies
 \begin{align}
   \tensor{T}{_\mu_{\s}^\mu} = \tensor{K}{_\mu_{\s}^\mu} &= \frac{2}{\rho}\del_\s\r \ .
 \label{tilde-T-T-contract}
 \end{align}
For the
rescaled frame resp. effective metric, we have
 \begin{align}
  \tensor{\cT}{_\mu_{\s}^\mu}  = \tensor{\cK}{_\mu_{\s}^\mu} 
  &= - \r^{-1} \del_\s\rho \ .
 \label{torsion-contorsion-contract}
\end{align}

\end{lem}

\begin{proof}
 
Using the above results, we can evaluate  the contraction of the torsion
\begin{align}
 \tensor{T}{_\mu_{\r}^\mu} &= \tensor{\G}{_\mu_{\r}^\mu} - \tensor{\G}{_{\r}_\mu^\mu} 
  =  \del_\r\big(\ln(\sqrt{|\gamma^{\mu\nu}|}) + \ln\r_M\big) \nn\\
   &= \del_\r\ln(\r_M\sqrt{|\gamma^{\mu\nu}|}) 
   = \del_\r\ln(\r^2) 
\end{align}
using $\sqrt{|\gamma^{\mu\nu}|} = \r^2 \r_M^{-1}$,
and 
\begin{align}
\tensor{K}{_{\mu}_{\nu}^{\mu}}
 = \frac 12 (\tensor{T}{_{\mu}_{\nu}^{\mu}} 
             + \tensor{T}{_{\mu}^{\mu}_{\nu}} 
             - \tensor{T}{_{\nu}^{\mu}_{\mu}})
 = \tensor{T}{_{\mu}_{\nu}^{\mu}}
\end{align}
using \eq{Levi-contorsion-basic}.
For the rescaled frame, \eq{eff-Gamma-rel} gives
\begin{align}    
  \tensor{\tilde\Gamma}{_{\mu}_\s^\mu}
 &=  \tensor{\Gamma}{_{\mu}_{\s}^\mu}
   + \r^{-1} \del_\s\rho  , 
\qquad \tensor{\tilde\Gamma}{_{\nu}_\mu^\mu}
 =  \tensor{\Gamma}{_{\nu}_{\mu}^\mu}
   + \frac{4}{\rho} \, \del_\nu\rho\  
%
 \label{Gamma-tilde-rel-contract}
\end{align}  
Therefore
\begin{align}
  \tensor{\cT}{_\mu_{\s}^\mu} 
  &=  \tensor{\tilde \G}{_\mu_{\s}^\mu} - \tensor{\tilde \G}{_{\s}_\mu^\mu} 
  =  \tensor{T}{_\mu_{\s}^\mu} - \frac{3}{\rho} \del_\s\rho  
  = - \r^{-1} \del_\s\rho 
  =  \tensor{\cK}{_\mu_{\s}^\mu}  \ . 
\end{align}

\end{proof}

\subsection{Equation of motion for the torsion}
\label{sec:conservation-law}

Now consider the equation of motion for torsion, which 
which results from the basic matrix equation of motion \eq{matrix-eom}
\begin{align}
 \{Z^{\dot\a},\hat\Theta_{{\dot\a}{\dot\b}}\} = m^2 Z_{\dot\b}
 \label{eom-YM-Theta}
\end{align}
in vacuum. Recalling  $\tensor{T}{_{\dot\a\dot\b}^\mu} = \{\hat\Theta_{\dot\a\dot\b},y^\mu\}$, we obtain the coordinate version from
\begin{align}
  m^2  \{Z_{\dot\b},y^\mu\}  &= \{\{Z^{\dot\a},\hat\Theta_{\dot\a\dot\b}\},y^\mu\} 
  = -  \{\{\hat\Theta_{\dot\a\dot\b},y^\mu\},Z^{\dot\a}\} - \{\{y^\mu,Z^{\dot\a}\},\hat\Theta_{\dot\a\dot\b}\} \nn\\
  &=  \{Z^{\dot\a},\tensor{T}{_{\dot\a\dot\b}^\mu}\} 
  + \{\tensor{E}{^{\dot\a}^\mu},\hat\Theta_{{\dot\a\dot\b}}\}  \ .
 \end{align}
The second term can be rewritten as  
\begin{align} 
 \{\tensor{E}{^{\dot\a}^\mu},\hat\Theta_{{\dot\a\dot\b}}\} 
 \sim  - \del_\nu \tensor{E}{^{\dot\a}^\mu} \tensor{T}{_{\dot\a\dot\b}^\nu}  
  = \tensor{\Gamma}{_\nu_\s^\mu} \tensor{E}{^{\dot\a}^\s} \tensor{T}{_{\dot\a\dot\b}^\nu} 
  = (\tensor{\Gamma}{_\s_\nu^\mu}  + \tensor{T}{_\nu_\s^\mu})\tensor{T}{^{\s}_{\dot\b}^\nu}  
  \end{align} 
 using the approximation \eq{poisson-approx-x} in the asymptotic regime.
 Furthermore, we note that
\begin{align}
 \nabla_{\nu} \tensor{T}{^{\nu}_{\dot\b}^\mu} 
 &= \nabla_{\nu} (\tensor{E}{_{\dot\a}^\nu}\tensor{T}{^{\dot\a}_{\dot\b}^\mu}) 
 = \tensor{E}{_{\dot\a}^\nu}\nabla_{\nu} \tensor{T}{^{\dot\a}_{\dot\b}^\mu} 
 = \tensor{E}{_{\dot\a}^\nu}
    (\del_{\nu} \tensor{T}{^{\dot\a}_{\dot\b}^\mu} 
    + \tensor{\Gamma}{_\nu_\s^\mu} \tensor{T}{^{\dot\a}_{\dot\b}^\s}) \nn\\
 &=  \{Z_{\dot\a}, \tensor{T}{^{\dot\a}_{\dot\b}^\mu}\}
   + \tensor{\Gamma}{_\s_\nu^\mu} \tensor{T}{^\s_{\dot\b}^\nu}
\end{align}
since $\nabla E = 0$.
Combining this we obtain
\begin{align}   
 m^2  \{Z_{\dot\b},y^\mu\} 
 &= \nabla_{\nu} \tensor{T}{^{\nu}_{\dot\b}^\mu} 
 + \tensor{T}{_\nu_{\s}^\mu} \tensor{T}{^{\s}_{\dot\b}^\nu} 
\end{align}
which amounts to
\begin{align}
 \nabla_\nu \tensor{T}{^\nu_\r^\mu}  +   \tensor{T}{_\nu_{\s}^\mu}\tensor{T}{^\s_\r^\nu} 
  = m^2 \d_\r^\mu \ .
\label{torsion-eq-nonlin-coord-1}
\end{align}
Upon lowering an index with $\g_{\mu\nu}$ we obtain
\begin{align}
\boxed{ \ 
 \nabla_\nu \tensor{T}{^\nu_\r_\mu}  +   \tensor{T}{_\nu^{\s}_\mu}\tensor{T}{_\s_\r^\nu} 
  = m^2 \g_{\r\mu} \  .
  }
\label{torsion-eq-nonlin-coord}
\end{align}
This non-linear equation encodes the non-linear structure of the Yang-Mills 
equations of motion \eq{eom-YM-Theta}.
Contracting with $\g^{\mu\r}$ gives 
\begin{align}
 \nabla_\nu \tensor{T}{^\nu_\mu^\mu}  + \tensor{T}{_\nu^{\s}_\mu}\tensor{T}{_\s_\r^\nu} \g^{\mu\r} = 4 m^2   \ .
\end{align}
The first term can be evaluated  using  \eq{tilde-T-T-contract} and \eq{Box-Box-relation}
and \eq{Gamma-contract-id} as
\begin{align}
 \nabla_\nu \tensor{T}{^\nu_\mu^\mu} 
  &= - 2\nabla_\nu (\r^{-1}\g^{\nu\s}\del_\s\r) 
  = 2 \r^{-2}\del_\nu\r\, \g^{\nu\s}\del_\s \r
   - 2 \r^{-1}\g^{\nu\s}\nabla_\nu\del_\s\r \nn\\
  &=  2 \r^{-2}\del_\nu\r\, \g^{\nu\s}\del_\s \r
   - 2 \r^{-1}\g^{\nu\s}(\del_\nu \del_\s\r- \tensor{\G}{_\nu_\s^\mu} \del_\mu\r) \nn\\
  &=  2 G^{\nu\s}\del_\nu\r\, \del_\s \r
   -2 \r G^{\nu\s}(\del_\nu \del_\s\r -  \tensor{\G}{^{(G)}_\nu_\s^\mu} \del_\mu\r) \nn\\
  &=  2 \big(G^{\mu\nu}\del_\mu\r\del_\nu \r + \r \Box_G \r\big) \ .
\end{align}
 We therefore obtain
\begin{align}
  \boxed{ \
  \r \Box_G \r + G^{\mu\nu}\del_\mu\r\del_\nu \r 
     = 2 m^2 - \frac 12 \tensor{T}{_\nu^{\s}_\mu}\tensor{T}{_\s_\r^\nu} \g^{\mu\r}
 \ . }
 \label{Box-r-onshell}
\end{align}
According to the discussion in section \ref{sec:poisson-red},
all these equations are exact for $\cC^0$, and asymptotic
for the higher spin components.
We verify in appendix \ref{sec:cosm-BG} that these equations are indeed satisfied 
exactly for the background solution and its torsion.
If desired,  \eq{torsion-eq-nonlin-coord}
 can be rewritten in terms of the Levi-Civita connection
using \eq{nabla-nablaLC-rel},
\begin{align}
 m^2 \g_{\r\mu} 
&= \nabla_\nu \tensor{T}{^\nu_\r_\mu}  +  \tensor{T}{_\nu^{\s}_\mu}\tensor{T}{_\s_\r^\nu} \nn\\
 &= \nabla^{(G)}_\nu \tensor{T}{^\nu_\r_\mu}  
   - \tensor{\cK}{_\nu_\r^\s}\tensor{T}{^\nu_\s_\mu} 
   - \tensor{\cK}{_\nu_\mu^\s}\tensor{T}{^\nu_\r_\s}
   + \tensor{T}{_\nu^{\s}_\mu}\tensor{T}{_\s_\r^\nu}  \ .
\label{torsion-eq-nonlin-coord-LC}
\end{align}
It could be expressed in terms of the torsion  only,
but this does not lead to a simpler expression.
Together with the equation \eq{Einstein-eq-vac} for the Ricci tensor and the Bianchi identity below, 
this provides a closed system of equations 
for the metric and the torsion.


\subsection{Bianchi identity}

There is  a Bianchi-type identity for the torsion, which results
from the Jacobi identity 
\begin{align}
 \{Z_{\dot\g},\hat\Theta_{{\dot\a}{\dot\b}}\} + ({\rm cycl}) =0
\end{align}
where $({\rm cycl})$ indicates cyclic permutations in ${\dot\a},{\dot\b},{\dot\g}$.
Its coordinate version is
\begin{align}
 0 &= \{\{Z_{\dot\g},\hat\Theta_{{\dot\a}{\dot\b}}\},y^\mu\} + ({\rm cycl}) \nn\\
  &= - \{\{\hat\Theta_{{\dot\a}{\dot\b}},y^\mu\},Z_{\dot\g}\} - \{\{y^\mu,Z_{\dot\g}\},\hat\Theta_{{\dot\a}{\dot\b}}\} + ({\rm cycl}) \nn\\
  &=  \{Z_{\dot\g},\tensor{T}{_{\dot\a}_{\dot\b}^\mu}\} 
   + \{\tensor{E}{_{\dot\g}^\mu},\hat\Theta_{{\dot\a}{\dot\b}}\} + ({\rm cycl})  \nn\\
  &\sim  \nabla_{\dot\g} \tensor{T}{_{\dot\a}_{\dot\b}^\mu} 
     - \tensor{\Gamma}{_{\dot\g}_\s^\mu} \tensor{T}{_{\dot\a}_{\dot\b}^\s}
  - \tensor{T}{_{\dot\a}_{\dot\b}^\nu} \del_\nu \tensor{E}{_{\dot\g}^\mu} + ({\rm cycl})  \nn\\
  &=  \nabla_{\dot\g} \tensor{T}{_{\dot\a}_{\dot\b}^\mu} 
   -  \tensor{T}{_{\dot\a}_{\dot\b}^\s}\tensor{\Gamma}{_{\dot\g}_\s^\mu}
   + \tensor{T}{_{\dot\a}_{\dot\b}^\nu} \tensor{\Gamma}{_{\nu}_{\dot\g}^\mu}  + ({\rm cycl})  \nn\\
  &=  \nabla_{\dot\g} \tensor{T}{_{\dot\a}_{\dot\b}^\mu} 
    + \tensor{T}{_{\dot\a}_{\dot\b}^\nu} \tensor{T}{_{\nu}_{\dot\g}^\mu}  + ({\rm cycl}) 
\end{align}
using again \eq{poisson-approx-x},
or equivalently 
\begin{align}
\boxed{\
 0 =  \nabla_\s \tensor{T}{_\l_\r^\mu}
 + \nabla_\l \tensor{T}{_\r_\s^\mu}  + \nabla_\r \tensor{T}{_\s_\l^\mu }
  + \tensor{T}{_\l_\r^\nu} \tensor{T}{_\nu_\s^\mu}
   + \tensor{T}{_\r_\s^\nu} \tensor{T}{_\nu_\l^\mu} 
   + \tensor{T}{_\s_\l^\nu} \tensor{T}{_\nu_\r^\mu}
 \ }
 \label{Bianchi-full}
\end{align}
which is cyclic in $\l,\r,\s$.
Together with \eq{torsion-eq-nonlin-coord} we obtained an analog of the Yang-Mills equations.
Contracting $\s\mu$, this gives 
\begin{align}
 0 &=  \nabla_\s \tensor{T}{_\l_\r^\s}  + \nabla_\l \tensor{T}{_\r_\s^\s }
 - \nabla_\r\tensor{T}{_\l_\s^\s} 
  -2 \r^{-1} \del_\nu \r  \tensor{T}{_\l_\r^\nu} 
\end{align}
using Lemma \ref{lem:contraction-torsion}. 
The middle terms can be rewritten using
\begin{align}
 (\nabla_\mu\del_\nu - \nabla_\nu\del_\mu)\r
 &= -\tensor{\Gamma}{_\mu_\nu^\s} \del_\s\phi + \tensor{\Gamma}{_\nu_\mu^\s} \del_\s\r
  = - \tensor{T}{_\mu_\nu^\s} \del_\s\r
\end{align}
and we obtain the identity
\begin{align}
\boxed{ \
\nabla_\s \tensor{T}{_\l_\r^\s}  = 0 \ .
 \ }
   \label{Bianchi-contract}
\end{align}
Contracting \eq{Bianchi-full} with $\g^{\s\l}$ does not give any non-trivial relation.

\subsection{Vacuum equation for the Ricci tensor}
\label{sec:Ricci-eq}

Now we compute the Ricci tensor for the Levi-Civita connection associated with the effective metric $G^{\mu\nu}$. 
This is achieved by expressing the Riemann tensor in terms of  the torsion.
We start from
\begin{align}
\tensor{\cR}{_\mu_\nu^\l_\s} 
  &= \del_\mu \tensor{\G}{^{(G)}_\nu_\s^\l} - \del_\nu \tensor{\G}{^{(G)}_\mu_\s^\l}
 + \tensor{\G}{^{(G)}_\mu_\r^\l} \tensor{\G}{^{(G)}_\nu_\s^\r}  
 - \tensor{\G}{^{(G)}_\nu_\r^\l} \tensor{\G}{^{(G)}_\mu_\s^\r}  \nn\\
 \cR_{\nu\s} &=\del_\mu \tensor{\G}{^{(G)}_\nu_\s^\mu} - \del_\nu \tensor{\G}{^{(G)}_\mu_\s^\mu}
 + \tensor{\G}{^{(G)}_\mu_\r^\mu} \tensor{\G}{^{(G)}_\nu_\s^\r}  
 - \tensor{\G}{^{(G)}_\nu_\r^\mu} \tensor{\G}{^{(G)}_\mu_\s^\r}  \ .
 \label{Riemann-tensor-gen}
\end{align}
In (Riemann) normal coordinates at $p\in\cM$, this simplifies using \eq{LC-contorsion-eff} as 
\begin{align}
\tensor{\cR}{_\mu_\nu^\l_\s} 
 &= \del_\mu (\tensor{\tilde\G}{_\nu_\s^\l} - \tensor{\cK}{_\nu_\s^\l}) 
   - \del_\nu (\tensor{\tilde\G}{_\mu_\s^\l} - \tensor{\cK}{_\mu_\s^\l}) \nn\\
 \cR_{\nu\s} &= \del_\mu (\tensor{\tilde\G}{_\nu_\s^\mu} - \tensor{\cK}{_\nu_\s^\mu}) 
              - \del_\nu (\tensor{\tilde\G}{_\mu_\s^\mu} - \tensor{\cK}{_\mu_\s^\mu}) \ .
\end{align}
Now we use the fact that the curvature of the Weitzenb\"ock connection vanishes,
\begin{align}
 0
 &= \del_\mu \tensor{\tilde\G}{_\nu_\s^\l} - \del_\nu\tensor{\tilde\G}{_\mu_\s^\l}
 +  \tensor{\tilde\G}{_\mu_\r^\l}  \tensor{\tilde\G}{_\nu_\s^\r}  
 -  \tensor{\tilde\G}{_\nu_\r^\l}  \tensor{\tilde\G}{_\mu_\s^\r}  \nn\\
 0 &=\del_\mu  \tensor{\tilde\G}{_\nu_\s^\mu} - \del_\nu  \tensor{\tilde\G}{_\mu_\s^\mu}
 +  \tensor{\tilde\G}{_\mu_\r^\mu}  \tensor{\tilde\G}{_\nu_\s^\r}  
 -  \tensor{\tilde\G}{_\nu_\r^\mu}  \tensor{\tilde\G}{_\mu_\s^\r}  
\end{align}
and  obtain the tensorial equations
\begin{align}
\tensor{\cR}{_\mu_\nu^\l_\s} 
 &= -\nabla^{(G)}_\mu \tensor{\cK}{_\nu_\s^\l}
   + \nabla^{(G)}_\nu  \tensor{\cK}{_\mu_\s^\l} 
   - \tensor{\cK}{_\mu_\r^\l} \tensor{\cK}{_\nu_\s^\r}  
   + \tensor{\cK}{_\nu_\r^\l} \tensor{\cK}{_\mu_\s^\r} \nn\\
 \cR_{\nu\s} &= -\nabla^{(G)}_\mu \tensor{\cK}{_\nu_\s^\mu}
                + \nabla^{(G)}_\nu \tensor{\cK}{_\mu_\s^\mu}  
        - \tensor{\cK}{_\mu_\r^\mu} \tensor{\cK}{_\nu_\s^\r}  
        + \tensor{\cK}{_\nu_\r^\mu} \tensor{\cK}{_\mu_\s^\r}       
  \label{riemann-K}
\end{align}
 using \eq{Weitzenbock-Torsion-RiemanNC}.
It remains to evaluate the derivative terms of the contorsion.
The first term can be evaluated using \eq{Levi-contorsion-full}, which gives
 \begin{align}
  \nabla^{(G)}_\mu\tensor{\cK}{_{\nu}_{\s}^\mu} 
   &= \nabla^{(G)}_\mu\tensor{K}{_{\nu}_\s^\mu}
   - G_{\nu\s} \big(\r^{-1}\Box_G \r + \r^{-2}\del\r\cdot\del\r\big)
   - \r^{-1}\nabla^{(G)}_\nu\del_\s\r 
    + \r^{-2}\del_\nu \r \del_\s\r  \ 
\end{align}
where
\begin{align}
 \del\r\cdot\del\r := G^{\mu\s}\del_\mu\r\del_\s \r \ .
\end{align}
Further, \eq{torsion-contorsion-contract} gives 
\begin{align}
 \nabla^{(G)}_\nu \tensor{\cK}{_\mu_\s^\mu} &= - \nabla^{(G)}_\nu(\r^{-1} \del_\s\rho )
  =  \r^{-2}\del_\nu\rho \del_\s\rho -\r^{-1} \nabla^{(G)}_\nu\del_\s\rho 
\end{align}
so that 
\begin{align}
 - \nabla^{(G)}_\mu\tensor{\cK}{_{\nu}_{\s}^\mu} + \nabla^{(G)}_\nu \tensor{\cK}{_\mu_\s^\mu}
  &= - \nabla^{(G)}_\mu\tensor{K}{_{\nu}_\s^\mu}
   + G_{\nu\s} \big(\r^{-1}\Box_G \r + \r^{-2}\del\r\cdot\del\r\big) \ .
\end{align}
The first term can be evaluated as
\begin{align}
 \nabla^{(G)}_\mu\tensor{K}{_{\nu}_\s^\mu} 
  &= \nabla_\mu\tensor{K}{_{\nu}_\s^\mu} 
  + \tensor{\cK}{_\mu_\nu^\r}  \tensor{K}{_\r_\s^\mu}  
  +  \tensor{\cK}{_\mu_\s^\r} \tensor{K}{_{\nu}_\r^\mu}  
\end{align}
using \eq{nabla-nablaLC-rel}, \eq{Levi-contorsion-full} and \eq{torsion-contorsion-contract}.
This gives after some straightforward algebra using \eq{tilde-K-K}
\begin{align}
  \cR_{\nu\s}
&=  - \nabla_\mu\tensor{K}{_{\nu}_\s^\mu} 
   - \tensor{K}{_\mu_\nu^\r} \tensor{K}{_\r_\s^\mu} 
   + 2\r^{-2}\del_\s \r \del_\nu \r   
   + G_{\nu\s} \big(\r^{-1}\Box_G \r + \r^{-2}\del\r\cdot\del\r\big) \ .
 \label{Ricci-aux1}
\end{align}
So far this is an identity.
We now replace the contorsion with the torsion using \eq{Levi-contorsion-basic}
and use the  Bianchi identity \eq{Bianchi-contract} as well as  the equation of motion \eq{torsion-eq-nonlin-coord} for the torsion.
This gives
\begin{align}
 \nabla_\mu \tensor{K}{_\nu_\s^\mu}
  &=  \frac 12 \nabla_\mu \big(\tensor{T}{_\nu_\s^\mu} + \tensor{T}{^\mu_\nu_\s}  
  - \tensor{T}{_\s^\mu_\nu} \big)  
  =  \frac 12 \nabla_\mu \big( \tensor{T}{^\mu_\nu_\s} + \tensor{T}{^\mu_\s_\nu} \big) \nn\\
  &= -  \frac 12 (\tensor{T}{^\r^{\eta}_\s}\tensor{T}{_\eta_\nu_\r}
     + \tensor{T}{^\r^{\eta}_\nu}\tensor{T}{_\eta_\s_\r})
     + m^2 \g_{\nu\s} \ ,
\end{align}
and we obtain the desired equation for the Ricci tensor of the effective metric in vacuum
\begin{align}
  \cR_{\nu\s}[G] &= - \frac 12 (\tensor{T}{_\r^{\mu}_\s}\tensor{T}{_\nu_\mu^\r} 
  + \tensor{T}{_\r^{\mu}_\nu}\tensor{T}{_\s_\mu^\r})
  - \tensor{K}{_\mu^\r_\nu}\tensor{K}{_\r^\mu_\s}
  + 2\r^{-2}\del_\nu\r \del_\s\r \nn\\
 &\quad  + G_{\nu\s} \big(-\r^{-2} m^2 + \r^{-1}\Box_G \r + G^{\mu\nu} \r^{-2}\del_\mu\r\del_\nu \r \big) \ .
   \label{Ricci-vacuum-1}
\end{align}
This agrees precisely with the result \eq{Ricci-alternative} obtained in a more pedestrian way.
The last term can be rewritten in terms of the torsion using the contracted equations of motion \eq{Box-r-onshell},
\begin{align}
  \cR_{\nu\s} &= - \frac 12 (\tensor{T}{_\r^{\mu}_\s}\tensor{T}{_\nu_\mu^\r} 
  + \tensor{T}{_\r^{\mu}_\nu}\tensor{T}{_\s_\mu^\r})
  - \tensor{K}{_\mu^\r_\nu}\tensor{K}{_\r^\mu_\s}
  + 2\r^{-2}\del_\nu\r \del_\s\r \nn\\
&\quad  + G_{\nu\s} \big(\r^{-2} m^2 - \frac 12 \tensor{T}{_\nu^{\s}_\mu}\tensor{T}{_\s_\r^\nu} G^{\mu\r} \big) \ .
 \label{Ricci-vacuum-2}
\end{align}
This is an algebraic equation for the Ricci tensor in terms of torsion\footnote{We can consider 
$\r^{-1}\del_\nu\r = \tensor{\cK}{_\mu^\mu_\nu}$ as part of the torsion \eq{torsion-contorsion-contract}.}.
The Ricci scalar is obtained by contracting  with $G^{\nu\s}$, 
\begin{align}
\cR[G] 
   &= -\tensor{T}{_\r^{\mu}^\nu}\tensor{T}{_\mu_\nu^\r} 
  - \tensor{K}{_\mu^\r_\nu}\tensor{K}{_\r^\mu^\nu} 
  +2 \r^{-2} \del\r\cdot\del\r + 4 \r^{-2} m^2  \nn\\
  &=  - \frac 12 \tensor{T}{^\s^\mu^\r}\tensor{T}{_\mu_\r_\s} 
 - \frac 14 \tensor{T}{^\mu^\s^\r}\tensor{T}{_\mu_\s_\r} 
  +2 \r^{-2} \big(\del\r\cdot\del\r + 2 m^2 \big)
\end{align}
using \eq{scalar-contraction-T} in the last step. 
Hence the present vacuum equations can be written as Einstein equations in the form
\begin{align}
\boxed{\ 
  \cG_{\mu\nu} =  \cR_{\mu\nu}  - \frac 12 G_{\mu\nu} \cR  = 8 \pi {\bf T}_{\mu\nu} 
\ }
\label{Einstein-eq-vac}
\end{align}
with an effective energy-momentum tensor due to the torsion,
\begin{align}
 8\pi {\bf T}_{\mu\nu} &= - \frac 12 (\tensor{T}{_\r^{\d}_\nu}\tensor{T}{_\mu_\d^\r} 
  + \tensor{T}{_\r^{\d}_\mu}\tensor{T}{_\nu_\d^\r})
  - \tensor{K}{_\d^\r_\mu}\tensor{K}{_\r^\d_\nu}
  + 2\r^{-2}\del_\mu\r \del_\nu\r \nn\\
 &\quad + G_{\mu\nu} \big(
 - \frac 14 \tensor{T}{^\s^\d^\r}\tensor{T}{_\d_\r_\s} 
 + \frac 18 \tensor{T}{^\d^\s^\r}\tensor{T}{_\d_\s_\r} 
  - \r^{-2} \del\r\cdot\del\r - 3R^{-2}\r^{-2} \big) 
  \label{em-tensor-effective}
\end{align}
recalling that $m^2 = 3R^{-2}$.
This is verified for the cosmic background in section \eq{sec:cosm-BG}. 
The conservation law $\nabla_{(G)}^\nu {\bf T}_{\nu\mu} = 0$ is guaranteed  
at least in vacuum because the rhs simply computes the Einstein tensor, which is conserved.

The above equations \eq{torsion-eq-nonlin-coord}, \eq{Bianchi-full} and \eq{Einstein-eq-vac} for the torsion and Ricci tensor 
provide a closed system 
of non-linear equations which govern the emergent gravity  on the present background.
The fact that the quantities are $\hs$-valued makes them rich and 
rather complicated\footnote{Recall that $\hs$ is a commutative in the 
semi-classical limit here, hence there is no ordering ambiguity.}.
Although the present derivation is restricted 
to the asymptotic regime for the $\hs$ sector, the equations 
are exact for the lowest components in $\cC^0$, as verified for the cosmological
background. Using covariance under higher-spin 
gauge transformations, it should be possible to largely transform away the $\hs$-components 
in many situations,
and the exact equations for the $\cC^0$ component should 
be accessible to analytic investigation.

\subsection{Discussion and further considerations}
\label{sec:discussion}

The crucial point of the above result is that ${\bf T}_{\mu\nu} = O(T T)$ 
is  quadratic in the torsion,
as appropriate for an energy-momentum tensor. 
In contrast, the Riemann tensor \eq{riemann-K} contains a (derivative) term which is linear 
in the torsion,
hence any non-trivial geometry has non-trivial torsion. 
However, this  linear contribution vanishes on-shell for the Ricci tensor, so
that vacuum geometries are Ricci-flat up to higher-order
(non-linear) contributions. This means that the present theory 
is  a serious candidate for gravity, and deviations from GR
(at least in vacuum) are restricted to the non-linear regime.

It remains to quantify the deviations of the present theory from GR. 
One regime where deviations will surely arise is the strong gravity regime, 
where the Riemann tensor is large in some sense. That is the regime where 
deviations from  GR typically arise in  alternative approaches to gravity such as  string theory.
However there is another regime where torsion may lead to significant modifications here, 
namely for very large, massive objects as discussed in the next section. 
This new mechanism arises from  the self-coupling of torsion due to \eq{torsion-eq-nonlin-coord}, 
and it might mimic the presence of dark matter as discussed below.

%
%
%
%

\paragraph{Torsion as a ``dark matter''.}

Consider first the linearization of the torsion 
\begin{align}
\tensor{T}{_{\r}_\nu^\mu}  = \tensor{\bar T}{_{\r}_\nu^\mu}  + \tensor{t}{_{\r}_\nu^\mu} 
\end{align}
around the cosmic background torsion $\bar T = O(\frac{1}{a(t)})$  \eq{torsion-BG-explicit}.
Since the  contributions from the background cancel, 
the eom \eq{torsion-eq-nonlin-coord-LC} takes the following schematic form
\begin{align}
 -\nabla_\nu^{(G)} \tensor{t}{^\nu_\r_\mu} 
  &\sim  - \frac{3}{a(t)^2} \d G_{\r\mu} 
  + \frac{1}{a(t)} (t)_{\r\mu}
  + (t t)_{\r\mu} \ .
\label{torsion-eq-lin-coord}
\end{align}
The constant and linear terms on the rhs are suppressed by the 
cosmic scale factor $\frac{1}{a(t)}$.
However, since the torsion is smaller than the cosmic background in the linearized regime, it will
not be  very significant physically. Therefore we need to consider the non-linear regime,
\begin{align}
 \nabla_\nu^{(G)} \tensor{t}{^\nu_\r_\mu}  &\sim -(t t)_{\r\mu} \ .
\label{torsion-eq-nonlin}
\end{align}
To get a rough qualitative idea what this could mean, 
consider the simplified radial equation
\begin{align}
 t' = - t^2
\end{align}
for $t = t(r)$, where $r$ is the distance to the center of some massive object.
This should give a reasonable picture for the radial dependence of the torsion tensor.
That equation has the general solution 
\begin{align}
 t(r) = \frac{1}{r+c}
 \label{torsion-radial-qual}
\end{align}
for some constant $c$. 
This leads to an effective energy-momentum tensor 
$8\pi {\bf T}_{\mu\nu} \sim \frac 1{(r+c)^2}$  \eq{em-tensor-effective} as a source of the  Einstein tensor,
which would behave like a dark matter halo with density profile
$\r_{\rm DM}(r)\sim \frac 1{(r+c)^2}$ corresponding to a total mass
\begin{align}
 M_{DM}(r) \sim  c + r  + O(\frac 1r) \ ,
 \qquad v_{\rm rot}(r) \sim {\rm const}
\end{align}
leading to a rotation velocity $v_{\rm rot}(r)$ which is roughly independent 
of the distance $r$ to the (galactic) center.
This is indeed what is typically observed. 
The scale parameter $c$ should be determined by continuity in a refined treatment, and 
 it is presumably set by the ``size'' or mass  of the object.
 At very large distances, the torsion \eq{torsion-radial-qual} will merge to that of the cosmic background,
 leading to a natural  cutoff for the effect.
 For small masses or objects, there will not be sufficient space for 
 \eq{torsion-radial-qual} to rise significantly above the background, 
 so that the effect should be significant only 
 for very large objects such as galaxies.

Needless to say that this crude qualitative consideration needs to be considerably refined before  quantitative 
statements can be made, and  the coupling to matter needs to be understood 
 and taken into account properly. 
 In any case, it is intriguing to 
obtain a qualitatively reasonable first estimate, and it is also encouraging that the 
present mechanism based on dynamical torsion is sufficiently rich that 
different types of behavior might be produced.

\paragraph{Coupling to matter.}

This paper is restricted to the vacuum geometry of the model. 
To properly talk about gravity we should of course take matter into account, which is 
 indeed an intrinsic part of the IKKT matrix model.
It is clear that the kinematics of matter is properly governed by the metric;
this is how the metric was identified\footnote{As discussed in \cite{Steinacker:2010rh}, 
the Dirac operator for fermions in the IKKT model is also based on the 
effective frame as it should. This should be elaborated 
in more detail elsewhere.}.
However, the non-trivial question is  how matter 
acts as a source for torsion and the Ricci tensor.
It is not clear what is the best way to work this out, and we
postpone this question to future work. However, a few comments can be made at this point:

First of all, due to the (higher-spin) covariance of the theory
it is highly plausible that the energy-momentum tensor for matter will arise on the rhs of the 
Einstein equations. However there will also be higher-derivative terms, and
the question is if the standard contribution dominates the
higher-derivative contributions; see also the related discussion in \cite{Sperling:2019xar}.
Covariance 
will strongly restrict the possible  terms, and it is certainly plausible that the 
energy-momentum tensor will dominate.
Due to the presence of several scales on the background
this must be studied in detail, in order to identify the effective Newton constant.
Since quantum effects are typically significant in this context, 
this may not be a trivial task.

Similarly, the effect of matter on torsion must be understood. This is expected to be small since 
for bosonic matter there should not source torsion at all (as only the metric appears in the kinetic term),
and for fermions the effect is expected to be small as well, due to 
the supersymmetry of the underlying matrix model.

\section{Conclusion and outlook}

The present paper provides a  tensorial description of the vacuum sector of the
effective gravity which arises on a  solution of the IKKT matrix models
found in \cite{Sperling:2019xar}, interpreted as FLRW space-time.
The noncommutative Yang-Mills gauge theory is cast into a geometric form which makes the  
gauge invariance manifest, in terms of a higher-spin 
generalization of volume-preserving diffeomorphisms. 
The crucial concept turns out to be torsion, or rather a higher-spin generalization of torsion,
which encodes the quantum structure of space-time and provides its semi-classical shadow.

Torsion turns out to be an independent and additional physical quantity besides the metric, 
and the Einstein equations for vacuum are modified through an effective
effective energy-momentum tensor due to torsion.
A non-linear equation for torsion is obtained, which encodes the underlying 
Yang-Mills-type equations of motion of the matrix model. 
This equation is exact for the standard (lowest-spin) tensorial components,
but obtained only in an intermediate (``asymptotic'') regime for the higher-spin components.

Moreover, we have argued that at least in the vacuum sector,
the modification of GR due to torsion should be small
except for very large objects such as galaxies, and on cosmic scales.
The point is that torsion enters  quadratically in the 
effective energy-momentum tensor, while it is governed itself by some non-linear PDE.
A  rough qualitative estimate suggests that it could indeed behave like an
apparent dark matter halo around galaxies. 
In principle,
the equations obtained in this paper should allow to obtain a quantitative 
description for this effect which can be tested.

The theory of gravity obtained in this way is governed by an action which is very different
from the Einstein-Hilbert action. 
Unlike in the teleparallel formulation of general relativity,
there is no  way to rewrite the matrix model as local action
in terms of the torsion, metric and frame.
It is precisely this non-standard non-geometric origin which makes the 
present approach to gravity so interesting and potentially far-reaching.

As in all higher-spin theories, an important question is if the theory reduces
in some suitable regime
to an ordinary (modified) gravity theory with spin $\leq 2$.  
The presence of both an IR scale (given by the cosmic curvature scale) 
{\em and} a UV scale (given by the scale of noncommutativity \eq{L-NC-def}) 
leads to the hope 
that this may be the case in the present framework. 
However,  this  needs to be studied in future work.

Moreover, the present paper is limited to the vacuum sector of gravity.
This restriction is only due to technical reasons, and
obviously needs to be removed in future work. In principle, 
everything should follow from the underlying matrix model,
and matter will certainly influence
the geometry in some way consistent with the (generalized) covariance. 
Moreover,
the framework of matrix models allows to make sense of the path integral.
In particular for the maximally supersymmetric IKKT model, one may reasonably hope that
the  present (semi-) classical treatment is not too far from 
the full quantum theory. 
The extra structure required for an interesting  matter sector 
can naturally arise from fuzzy extra dimensions realized by the 
extra 6 bosonic matrices in the model, as discussed e.g. in 
\cite{Chatzistavrakidis:2011gs,Aschieri:2006uw,Sperling:2018hys,Aoki:2014cya},
see also \cite{Hatakeyama:2019jyw}.

Finally, a general message  is that we ought to be   
cautious in extrapolating general relativity to regimes where it was not directly tested. 
The present theory may reproduce GR quite well in intermediate regimes
and it has a healthy linear excitation spectrum without
ghosts \cite{Steinacker:2019awe}, 
but it  certainly differs significantly on very long scales. Thus
the correct theory of gravity may be far richer than GR, and
the  puzzles of dark matter,  dark energy and the cosmological constant
may well  be evidence supporting such a picture.

\paragraph{Acknowledgments.}

I would like to thank Loriano Bonora and Stefan Fredenhagen for useful discussions,
and Marcus Sperling for collaboration in the early stages of this project.
This work was supported by the Austrian Science Fund (FWF) grant P32086.

\section{Technical supplements}

\subsection{Volume-preserving diffeomorphisms}
\label{sec:vol-pres-diffeo}

Consider the vector field \eq{diffeo-Lambda} associated a gauge transformation 
generated by  $\L^{(s)}\in\cC^s$
\begin{align}
 \xi^\mu := \{\L^{(s)},x^\mu\}  = \xi^\mu_+ + \xi^\mu_-  \qquad  \in \cC^{s+1} \oplus \cC^{s-1}
 \label{rank-1-field-def-2}
\end{align}
Using the basic identities (A.35) in \cite{Sperling:2019xar},
it is easy to see that 
these components  satisfy the following constraint 
\begin{align}
 \{t_{\mu}, \xi^\mu_+  \} &= -\frac{s+3}{R^2\sinh(\eta)}  x_{\mu}\xi^\mu_+ , \qquad
 \{t_{\mu}, \xi^\mu_- \} = -\frac{-s+2}{R^2\sinh(\eta)}  x_{\mu}\xi^\mu_- 
\end{align}
or equivalently
\begin{align}
  \bar\nabla_\mu(\b^{s+5} \xi^\mu_+) = 0 = \bar\nabla_\mu(\b^{-s+4} \xi^\mu_-), 
  \qquad \b = \frac 1{\sinh(\eta)} \ .
\end{align}
Here $\bar\nabla$ is the Levi-Civita derivative w.r.t. 
the cosmic background metric \cite{Steinacker:2019dii}.
In this sense, $\xi^\mu$
can be interpreted as volume-preserving higher spin diffeomorphism.

\subsection{Calculations for the Poisson bracket}
\label{sec:poisson-deriv-formulas}

We claim that
the following formula realizes the ansatz \eq{poissonbracket-Ansatz-coords} 
for the Poisson brackets:
\begin{align}
 \cosh^2(\eta) \{f,g\} 
 &=  \Big(\sinh(\eta)\{t_\mu,f\}
  - \frac{1}{r^2 R^2} \{f,x_\nu\}\big(\theta^{\nu\mu} 
     + r^2\sinh^{-1}(\eta)(t^\nu x^\mu - x^\nu t^\mu)\big)\Big)\{x_\mu,g\} \nn\\
 &\quad + \Big(\sinh(\eta)\{f,x^\mu\} + \{f,t_\nu\}\theta^{\nu\mu}\Big)\{t_\mu,g\} \ .
 \label{Poisson-full-formula}    
  %
\end{align}
The two terms look different due to the ambiguity in \eq{poissonbracket-Ansatz-coords},
and different forms can be obtained using
\begin{align}
  \{f,t_\nu\}(t^\nu x^\mu - x^\nu t^\mu)\{t_\mu,g\} 
   &= \{f,t_\nu\}t^\nu x^\mu \{t_\mu,g\} - \{f,t_\nu\}x^\nu t^\mu\{t_\mu,g\} \nn\\
   &= \frac{1}{r^2 R^2}\{f,x_\nu\}(x^\nu t^\mu - t^\nu x^\mu)\{x_\mu,g\} \ .
\end{align}
This leads to the following closed formulas for the derivatives $\eth$ and $\cD$
in \eq{derivations-coords}:
\begin{align}
\cosh^2(\eta) \eth_\mu f &=\sinh(\eta)\{t_\mu,f\}
  - \frac{1}{r^2 R^2} \big(\theta^{\nu\mu} 
     + r^2\sinh^{-1}(\eta)(t^\nu x^\mu - x^\nu t^\mu)\big)\{f,x_\nu\} \nn\\
 \cosh^2(\eta)\cD^\mu(f) &= \sinh(\eta)\{f,x^\mu\} + \theta^{\nu\mu} \{f,t_\nu\} \ .
 \label{eth-del-explicit}
\end{align}
For the generators, this gives
\begin{align}
 \cosh^2(\eta) \cD^\mu t_\r &= \sinh(\eta)\{t_\r,x^\mu\} + \{t_\r,t_\nu\}\theta^{\nu\mu}  \nn\\
  &= \sinh^2(\eta)\d_\r^\mu -\frac{1}{r^2 R^2}\theta^{\r\nu}\theta^{\nu\mu}  \nn\\
  &= \cosh^2(\eta) P_\perp^{\r\mu}  - r^2 t^\r t^\mu , \qquad P_\perp^{\r\mu} = \eta^{\r\mu} -\frac{1}{x_\nu x^\nu}x^\r x^\mu  \nn\\
  %
 \cosh^2(\eta) \cD^\mu x_\r &= \sinh(\eta)\theta^{\r\mu} -\sinh(\eta) \theta^{\r\mu} = 0 \nn\\
  \cosh^2(\eta)\eth^\mu x_\r &= \sinh(\eta)\{t_\mu,x_\r\}
  - \frac{1}{r^2 R^2} \{x_\r,x_\nu\}\big(\theta^{\nu\mu} 
     + r^2\sinh^{-1}(\eta)(t^\nu x^\mu - x^\nu t^\mu)\big)  \nn\\
 &= \sinh^2(\eta)\d_\mu^\r
  - \frac{1}{r^2 R^2} \big(\theta^{\r\nu}\theta^{\nu\mu} 
     + r^2\sinh^{-1}(\eta)(\theta^{\r\nu}t^\nu x^\mu - \theta^{\r\nu}x^\nu t^\mu)\big)    \nn\\
 &= \cosh^2(\eta) \eta^{\r\nu}   \nn\\
 \cosh^2(\eta)\eth^\mu t_\r  &= \sinh(\eta)\{t_\mu,t_\r\}
  - \frac{1}{r^2 R^2} \{t_\r,x_\nu\}\big(\theta^{\nu\mu} 
     + r^2\sinh^{-1}(\eta)(t^\nu x^\mu - x^\nu t^\mu)\big)  \nn\\
 &= -\frac{1}{r^2 R^2}\sinh(\eta)\theta^{\mu\r}
  - \frac{1}{r^2 R^2} \sinh(\eta)\big(\theta^{\r\mu} 
     + r^2\sinh^{-1}(\eta)(t^\r x^\mu - x^\r t^\mu)\big) \nn\\
 &= - \frac{1}{R^2}(t^\r x^\mu - x^\r t^\mu)
\end{align}
One can then verify  \eq{Poisson-full-formula} and \eq{poisson-bracket-realiz} explicitly.

\paragraph{Locally rescaled generators.}

To get a better intuition for the Poisson brackets, let us
define  adapted momentum generators for some cosmic time scale $\eta_0$ near some observer,
\begin{align}
 p_\a := \b_0 t_\a, \qquad  p_\a p^\a \approx r^{-2} 
\end{align}
where $\b_0 = \cosh^{-1}(\eta_0) \ll 1$.
They satisfy approximately canonical commutation relations
\begin{align}
 \{p_\a,  x^\mu\} &= \b_0 \sinh(\eta) \d_\a^\mu  \ \approx \ \d_\a^\mu \ .
\end{align}
The noncommutativity of the remaining generators is obtained from \eq{CR-explicit-ref}, 
\begin{align}
 \{p^0, p^j\} &= - \frac{\b_0}R  p^j  \nn\\
 \{p^i, p^j\} &= \frac{\b_0^2}{R}  \epsilon^{ijk} p^k   \nn\\
 \{x^0, x^j\} &= r^2 R \b_0^{-1} p^j \nn\\
 \{x^i, x^j\} &= r^2 R \varepsilon^{ijk} p^k
\end{align}
at the reference point.
Hence $\{p^\mu, p^\nu\} = O(\frac 1{L_{NC}^2})$
while $\{x^\mu, x^\nu\} = O(L_{NC}^2)$.

\subsection{Proof of Lemma \ref{lem:Box}}
\label{sec:proof-Lemma-Box}

Consider first the following contraction
\begin{align}
 \tensor{\G}{_{\dot\g}_\mu^\mu} &= 
  \tensor{\G}{_{\dot\g}_{\dot\b}^\mu}  \tensor{E}{^{\dot\b}_\mu} 
   = -\{Z_{\dot\g},\tensor{E}{_{\dot\b}^\mu}\} \tensor{E}{^{\dot\b}_\mu }
   = -\{Z_{\dot\g},\tensor{E}{_{\dot\b}^\mu}\} \tensor{E}{_{\dot\b}^\nu} {\g}_{\nu\mu} \nn\\
  &= -\{Z_{\dot\g},\tensor{E}{_{\dot\b}^\mu} \tensor{E}{_{\dot\b}^\nu}\}  {\g}_{\nu\mu} 
     +\{Z_{\dot\g}, \tensor{E}{_{\dot\b}^\nu}\} \tensor{E}{_{\dot\b}^\mu} {\g}_{\nu\mu} \nn\\
  &= -\{Z_{\dot\g},{\g}^{\mu\nu}\} {\g}_{\nu\mu} + \{Z_{\dot\g}, \tensor{E}{_{\dot\b}^\nu}\} \tensor{E}{_{\dot\b}^\mu} {\g}_{\nu\mu} \nn\\
  &= -\{Z_{\dot\g},{\g}^{\mu\nu}\} {\g}_{\nu\mu} + \{Z_{\dot\g}, \tensor{E}{_{\dot\b}^\nu}\} \tensor{E}{^{\dot\b}_{\nu}} \nn\\
  &= -\{Z_{\dot\g},{\g}^{\mu\nu}\} {\g}_{\nu\mu} - \tensor{\G}{_{\dot\g}_\mu^\mu} 
\end{align}
This gives \eq{Gamma-det-identity}
\begin{align}
 2 \tensor{\G}{_{\dot\g}_\mu^\mu} &= -\{Z_{\dot\g},\g^{\mu\nu}\} \g_{\nu\mu}
  = -\tensor{E}{_{\dot\g}^\r}\, \g_{\nu\mu} \del_\r\g^{\mu\nu} 
   = -\tensor{E}{_{\dot\g}^\r}\, \del_\r\ln(|\gamma^{\mu\nu}|)  \nn\\
   \tensor{\G}{_{\r}_\mu^\mu} &=  \del_\r\ln(\sqrt{|\gamma_{\mu\nu}|}) \ ,
\end{align}
which has the same form as for the Christoffel symbols for the Levi-Civita connection.
 Now consider the d'Alembertian $\Box \phi$:
 \begin{align}
  \{Z_{\dot\a},\{Z^{\dot\a},\phi\}\} &= \tensor{E}{_{\dot\a}^\mu}\del_\mu(\tensor{E}{^{\dot\a}^\nu}\del_\nu\phi)  \nn\\
  &= \del_\mu(\tensor{E}{_{\dot\a}^\mu}\tensor{E}{^{\dot\a}^\nu}\del_\nu\phi)  
    - \del_\mu(\tensor{E}{_{\dot\a}^\mu})\tensor{E}{^{\dot\a}^\nu}\del_\nu\phi  \nn\\
&= \del_\mu(\g^{\mu\nu}\del_\nu\phi)  
    + \tensor{\G}{_\r_{\dot\a}^\r} \tensor{E}{^{\dot\a}^\nu}\del_\nu\phi \nn\\
&= \del_\mu(\g^{\mu\nu}\del_\nu\phi)   + \tensor{\G}{_\mu^\nu^\mu} \del_\nu\phi  \ .
\label{Box-id-1}
 \end{align}    
Alternatively, we can proceed as      
\begin{align}
  \{Z_{\dot\a},\{Z^{\dot\a},\phi\}\}   
      &= \g^{\mu\nu}\del_\mu\del_\nu\phi + \tensor{E}{_{\dot\a}^\mu}(\del_\mu\tensor{E}{^{\dot\a}^\nu})\del_\nu\phi \nn\\
   &= \g^{\mu\nu}\del_\mu\del_\nu\phi - \tensor{E}{_{\dot\a}^\mu}\tensor{\G}{_\mu_\r^\nu} \tensor{E}{^{\dot\a}^\r}\del_\nu\phi \nn\\
   &= \g^{\mu\nu}\del_\mu\del_\nu\phi - \tensor{\G}{_\mu_\r^\nu} \g^{\r\mu}\del_\nu\phi \nn\\
   &= \g^{\mu\nu}\del_\mu\del_\nu\phi - \tensor{\G}{_\mu^\mu^\nu}\del_\nu\phi \ .
\label{Box-Poisson-2}
\end{align}
In particular, this gives \eq{Gamma-contract-id}
\begin{align}
 \Box y^\s &= \del_\mu\g^{\mu\s}
    + \tensor{\G}{_\r^\s^\r}  
   = -\tensor{\G}{_\mu^\mu^\s} \ .
\end{align}
As a check, the second line in \eq{LC-contorsion-1} reproduces the standard identity
in Riemannian geometry
\begin{align}
  \tensor{\G}{^{(\g)}_\mu^\mu^\s}
  &= \tensor{\G}{_\mu^\s^\mu} + \tensor{\G}{_\mu^\mu^\s} - \tensor{\G}{^\s^\mu_\mu} \nn\\
  &=  - \del_\mu\g^{\mu\s} +  \del^\s\ln(\sqrt{|\gamma^{\mu\nu}|})  
  = - \frac{1}{\sqrt{|\gamma_{\mu\nu}|}}\del_\mu\big(\sqrt{|\gamma_{\mu\nu}|}\g^{\mu\s}\big) \ .
\end{align}
To show \eq{Gamma-mu-t-mu-id}, 
we observe that  the basic frame can be written using \eq{reduced-bracket} as follows
\begin{align}
 \tensor{E}{_{\dot\b}^\mu} = \{Z_{\dot\b},y^\mu\} = - \theta^{\mu\s}\del_\s Z_{\dot\b}
\end{align}
Therefore
\begin{align}
 \del_\mu(\r_M\tensor{E}{_{\dot\b}^\mu}) = -\del_\mu(\r_M\theta^{\mu\s}\del_\s Z_{\dot\b})
  = -\del_\mu(\r_M\theta^{\mu\s})\del_\s Z_{\dot\b} = 0
\end{align}
using the identity $\del_\mu(\r_M\theta^{\mu\s}) = 0$ \eq{Poisson-conservation}
which follows from the Jacobi identity.
Together with $\del_\mu \tensor{E}{_{\dot\b}^\mu} = - \tensor{\G}{_\mu_{\dot\b}^\mu}$ 
this gives
\begin{align}
 \r_M \tensor{\G}{_\r_{\dot\b}^\r} &= \tensor{E}{_{\dot\b}^\mu}\del_\mu \r_M 
\end{align}
which gives \eq{Gamma-mu-t-mu-id}.
We can use this to continue \eq{Box-id-1} as follows 
\begin{align}
  \{Z_{\dot\a},\{Z^{\dot\a},\phi\}\} 
 &= \del_\mu(\g^{\mu\nu}\del_\nu\phi)   + \tensor{\G}{_\mu^\nu^\mu} \del_\nu\phi   \nn\\
 &= \del_\mu(\g^{\mu\nu}\del_\nu\phi)   + \r_M^{-1} \g^{\mu\nu}\del_\mu \r_M \del_\nu\phi   \nn\\
 &= \r_M^{-1}\del_\mu(\r_M\g^{\mu\nu}\del_\nu\phi)     \nn\\
 &= \frac{\r^2}{\sqrt{|G_{\mu\nu}|}}\del_\mu(\sqrt{|G_{\mu\nu}|}G^{\mu\r}\del_\nu\phi)     
\end{align}
using $\sqrt{|G_{\mu\nu}|}G^{\mu\r} = \r_M \g^{\mu\nu}$ 
and $\sqrt{|G_{\mu\nu}|} = \r_M \r^2$.
Now \eq{Box-Box-relation} follows.

Finally, contacting \eq{weizenbock-gamma-expl} with $\tensor{E}{^{\dot\a}^\nu}$
gives
\begin{align}
 \tensor{E}{^{\dot\a}^\nu}\del_\nu\tensor{E}{_{\dot\a}^\mu} \ 
 &= -\tensor{\Gamma}{_{\nu}_\r^\mu}\tensor{E}{_{\dot\a}^\r} \tensor{E}{^{\dot\a}^\nu} 
 = -\tensor{\Gamma}{_{\nu}_\r^\mu} \g^{\nu\r}
\end{align}
hence 
\begin{align}
 \{Z^{\dot\a},\{Z_{\dot\a},y^\mu\}\} &= -\tensor{\Gamma}{_{\nu}_\r^\mu} \g^{\nu\r} \nn\\
  \Box y^\mu &= \tensor{\Gamma}{_{\nu}_\r^\mu} \g^{\nu\r} \ .
\end{align}

\subsection{Weitzenb\"ock connection for the cosmic background}
\label{sec:cosm-BG}

For the cosmic background $\cM$ defined by 
$Z_{{\dot\a}} = T_{{\dot\a}} \sim t_{{\dot\a}}$ \eq{T-solution}, we have 
in  Cartesian coordinates $x^\mu$ 
\begin{align}
 \tensor{\bar E}{^{\dot\a}^\mu} &= \{t^{\dot\a},x^\mu\} = \sinh(\eta)\eta^{\dot\a \mu} ,&
 \qquad {\obar\g}^{\mu\nu} &= \sinh^2(\eta) \eta^{\mu\nu}  \nn\\
  \bar G^{\mu\nu} &=  \sinh^{-3}(\eta)\bar\g^{\mu\nu} = \sinh^{-1}(\eta) \eta^{\mu\nu}  ,&
  \qquad\quad \bar\r^2 &= \sinh^{3}(\eta) \ . 
   \label{eff-metric-G}
\end{align}
The Weitzenb\"ock connection is obtained as
\begin{align}
\tensor{\Gamma}{_{\nu}_\r^\mu} &= -\tensor{E}{^{\dot\a}_\r} \del_\nu\tensor{E}{_{\dot\a}^\mu} 
 = -\sinh^{-1}(\eta)\d^\mu_\r \del_\nu  \sinh(\eta)  \nn\\
 \tensor{\Gamma}{^{\nu}_\r_\mu} 
  &= \frac{1}{R^2\r^2}\t^\nu \bar G_{\mu\r} 
   = \ \frac{1}{R^2}\t^\nu \bar \g_{\mu\r}  , \qquad \t = x^\mu\del_\mu
  \label{Gamma-BG-explicit}
\end{align}
noting that $R^2\sinh(\eta)\del_\mu \sinh(\eta) = -\eta_{\mu\nu}\t^\nu$.
Here $\t = x^\mu\del_\mu$ is the $SO(3,1)$-invariant cosmic time-like vector field (in Cartesian coordinates)
which measures the cosmic scale,
\begin{align}
 \bar G_{\r\r'}\t^\r\t^{\r'} &= -R^2 \cosh^2(\eta)\sinh(\eta) = - a(t)^2 \ \sim - R^2 \bar\r^2
 \label{tau-length}
\end{align}
where $a(t)$ is the FLRW scale factor, cf. \cite{Steinacker:2019dii}.
Hence the torsion is  
\begin{align}
 \tensor{\bar T}{_{\r}_{\s}^\mu} 
  =  \frac{1}{R^2\r^2}\big( \d_\s^{\mu} \t_\r - \d^\mu_\r \t_\s \big) \ 
  \label{torsion-BG-explicit}
\end{align}
where $\t_\nu = \bar G_{\nu\s}\t^\s$.
As a check, 
the torsion tensor can also be computed directly from \eq{torsion-explicit}.
One can also verify 
\begin{align}
 \tensor{\bar T}{_{\mu}_{\s}^\mu} 
  =  2\r^{-1} \del_\s \bar\r = - \frac{3}{R^2\r^2} \t_\s \ .
\end{align}

The contorsion is
\begin{align}
 \tensor{\bar K}{_{\mu}_{\nu}^{\s}}
 &= \frac 1{R^2\r^2} ( \bar G_{\mu\nu} \t^\s  -  \d_\mu^{\s} \t_\nu ) 
\end{align}
which is antisymmetric in $\nu\s$.
We also need 
\begin{align}
 \tensor{\bar T}{_\r^{\mu}_\s}\tensor{\bar T}{_\nu_\mu^\r} 
  &= \frac 1{R^{4}\r^4}( - \t_\s  \t_\nu + \bar G_{\s\nu}  \t^\mu \t_\mu) = \tensor{\bar T}{_\r^{\mu}_\nu}\tensor{\bar T}{_\s_\mu^\r} \nn\\ 
  %
 \tensor{\bar K}{_\mu^\r_\nu}\tensor{\bar K}{_\r^\mu_\s}
  &= \frac 3{R^{4}\r^4} \t_\nu \t_\s \nn\\
 -\frac 12(\tensor{\bar T}{_\r^{\mu}_\s}\tensor{\bar T}{_\nu_\mu^\r} 
  + \tensor{\bar T}{_\r^{\mu}_\nu}\tensor{\bar T}{_\s_\mu^\r})
  - \tensor{\bar K}{_\mu^\r_\nu}\tensor{\bar K}{_\r^\mu_\s}
  &= -\frac 1{R^{4}\r^4}(2\t_\s  \t_\nu + \bar G_{\s\nu}  \t^\mu \t_\mu) \nn\\
  2 \bar\r^{-2} \del_\s \bar\r \del_\nu \bar\r &=  \frac{9}{2R^4\r^4} \t_\s \t_\nu  \nn\\
  \r\Box_G\r &= \frac{3}{2R^2}\Big(4 + \frac 12 \frac{\cosh^2(\eta)}{\sinh^2(\eta)}\Big)  
\end{align}
Now we check  the eom \eq{torsion-eq-nonlin-coord}
for the background.
Using 
\begin{align}
 \nabla_\nu \t^\r &= \del_\nu \t^\r + \tensor{\Gamma}{_{\nu}^\r_\mu} \t^\mu
  = \d_\nu^\r +  \frac{1}{R^2\r^2} \t_\nu \t^\r 
 \label{nabla-tau}
\end{align}
one finds
\begin{align}
 \nabla_\nu \tensor{\bar T}{^\nu_\r_\mu} 
  &= \frac{1}{R^2\r^2}\big(\bar G_{\r\mu} (3 +  \frac{1}{R^2\r^2} \t_\nu \t^\nu) 
    - \frac{1}{R^2\r^2} \t_\mu \t_\r \big)
\end{align}
so that  \eq{torsion-eq-nonlin-coord} is indeed satisfied,
\begin{align}
 \nabla_\nu \tensor{\bar T}{^\nu_\r_\mu} + \tensor{\bar T}{_\nu^{\s}_\mu}\tensor{\bar T}{_\s_\r^\nu} 
  = \frac{3}{R^2\r^2} \bar G_{\r\mu}
  = m^2 \bar \g_{\r\mu} \ 
\end{align}
using $m^2 = \frac{3}{R^2}$. 
We can also  check \eq{Box-r-onshell}.
The Ricci tensor is obtained from \eq{Ricci-vacuum-2} as
\begin{align}
 \cR_{\mu\nu} &=  - \frac 12 (\tensor{\bar T}{_\r^{\mu}_\s}\tensor{\bar T}{_\nu_\mu^\r} 
  + \tensor{\bar T}{_\r^{\mu}_\nu}\tensor{\bar T}{_\s_\mu^\r})
  - \tensor{\bar K}{_\mu^\r_\nu}\tensor{\bar K}{_\r^\mu_\s}
  + 2\bar \r^{-2}\del_\nu\bar \r \del_\s\bar \r 
   + \g_{\nu\s} \big( m^2 - \frac 12 \tensor{T}{_\nu^{\s}_\mu}\tensor{T}{_\s_\r^\nu} \g^{\mu\r}  \big)   \nn\\
  %
   &=   \frac{5}{2}\frac{1}{\r^{4}R^4} \t_\s \t_\nu  + \frac{1}{2\r^2R^2} G_{\nu\s} (6 - \coth^2(\eta)) \nn\\
   &\sim  \frac 52\frac 1{a(t)^2} \Big(\frac{1}{a(t)^2} \t_\s \t_\nu  +  G_{\nu\s}\Big) \ .
\end{align}
The second line is indeed the exact result as can be checked with more standard methods,
and the third line holds in the asymptotic regime for large $\eta$.
This method of computing the Ricci tensor using the torsion is in fact quite efficient.
Hence we obtain the  effective vacuum energy-momentum tensor \eq{em-tensor-effective} 
\begin{align}
 {\bf \bar T}_{\mu\nu} = \cR_{\mu\nu} - \frac 12 G_{\mu\nu} \cR  
 &=  \frac 1{\r^2R^2}\Big(\frac{5}{2} \frac{\t_\mu\t_\nu}{\r^2 R^2} 
 - \frac 12 G_{\mu\nu} \big(-\frac{5}{2}\frac{a(t)^2}{\r^2R^2}  
 + (6 - \coth^2(\eta))\big)\Big) \nn\\
 &\sim \frac{5}{2} \frac 1{a(t)^2}\Big(\frac{\t_\mu \t_\nu}{a(t)^2}  - \frac 12 G_{\mu\nu}\Big) \ .
\end{align}
In  comoving coordinates, this has the form ${\bf \bar T}_{\mu\nu} \sim \frac 1{a(t)^2}\diag(3,-1,-1,-1)$.
Hence pressure is negative with $\omega = \frac{p}{\r} \sim -\frac 13$, 
and the strong energy condition is (just) satisfied.

\subsection{Alternative derivation of the Ricci tensor in vaccum} 
\label{sec:alt-Ricci}

Here we give a direct derivation of the on-shell equation of the Ricci tensor in a more 
matrix-model-adapted approach,
using only the matrix equations of motion rather than the derived equation \eq{torsion-eq-nonlin-coord} for the torsion.
The result agrees perfectly with \eq{Ricci-vacuum-1}, which provides a consistency check of the present formalism.
We start with 
\begin{align}
 \{y^\mu,\Box\phi\} &= -  \{y^\mu,\{Z^{\dot\a},\{Z_{\dot\a},\phi\}\} 
  = -  \{\{y^\mu,Z^{\dot\a}\},\{Z_{\dot\a},\phi\}\} - \{Z^{\dot\a}, \{y^\mu,\{Z_{\dot\a},\phi\}\}  \nn\\
  &= \{\tensor{E}{^{\dot\a}^\mu},\{Z_{\dot\a},\phi\}\} - \{Z^{\dot\a}, \{\{y^\mu,Z_{\dot\a}\} ,\phi\}\}
      - \{Z^{\dot\a},\{Z_{\dot\a}, \{y^\mu,\phi\}\} \nn\\
  &= \{\tensor{E}{^{\dot\a}^\mu},\{Z_{\dot\a},\phi\}\} +  \{Z^{\dot\a}, \{\tensor{E}{_{\dot\a}^\mu},\phi\}\}  + \Box(\{y^\mu,\phi\}\nn\\
  &= -\{\{\tensor{E}{^{\dot\a}^\mu},Z_{\dot\a}\},\phi\}\} + 2 \{\tensor{E}{_{\dot\a}^\mu},\{Z_{\dot\a},\phi\}\}  + \Box(\{y^\mu,\phi\}
\end{align}
 for any scalar field $\phi$.
Now assume that $y^\mu$ are harmonic coordinates, 
\begin{align}
 \Box_G y^\mu = 0 = \Box y^\mu
\end{align}
recalling \eq{Box-Box-relation}.
It follows that
\begin{align}
 \{\tensor{E}{_{\dot\a}^\mu},Z^{\dot\a}\} &= \{\{Z_{\dot\a},y^\mu\},Z^{\dot\a}\} = \Box y^\mu = 0
 \label{harmonic-coords-2}
\end{align}
and we obtain
\begin{align}
 \Box\{y^\mu,\phi\} &=  \{y^\mu,\Box\phi\} - 2 \{\tensor{E}{^{\dot\a}^\mu},\{Z_{\dot\a},\phi\}\} \ .
 \label{intertwiner-nonlin}
\end{align}
Now we use the equation of motion \eq{matrix-eom} for $Z^{\dot\a}$,
\begin{align}
 \Box Z^{\dot\a} = m^2 Z^{\dot\a}
\end{align}
Then \eq{intertwiner-nonlin} gives for $\phi = Z^{\dot\b}$
\begin{align}
\Box\{y^\mu,Z^{\dot\b}\} &=  \{y^\mu,\Box Z^{\dot\b}\} - 2 \{\tensor{E}{^{\dot\a}^\mu},\{Z_{\dot\a},Z^{\dot\b}\}\} \nn\\
 &= m^2 \{y^\mu,Z^{\dot\b}\} + 2 \{\tensor{E}{_{\dot\a}^\mu},\hat\Theta^{{\dot\a}{\dot\b}}\}  \nn\\
 \Box\tensor{E}{_{\dot\b}^\mu} &= m^2 \tensor{E}{_{\dot\b}^\mu} - 2 \{\hat\Theta^{{\dot\b}{\dot\a}},\tensor{E}{_{\dot\a}^\mu}\} \ .
  \label{Box-vielbein-id} 
\end{align}
For the rescaled frame $\cE^\mu_{\dot\a} = \r^{-1} E^\mu_{\dot\a}$ \eq{rescaled-frame-def}, this gives 
\begin{align}
 \Box \tensor{\cE}{_{\dot\a}^\mu} 
 &= (m^2 - \r^{-1}\Box\r) \tensor{\cE}{_{\dot\a}^\mu}
   -2 \{\hat\Theta^{{\dot\a}{\dot\b}},\tensor{\cE}{_{\dot\b}^\mu}\}  
  - 2 \r^{-1}\{\hat\Theta^{{\dot\a}{\dot\b}},\r\} \tensor{\cE}{_{\dot\b}^\mu} 
  + 2\{Z^{\dot\g}, \tensor{\cE}{_{\dot\a}^\mu}\}\r^{-1}\{Z_{\dot\g},\r\} \ .
  \label{Box-vielbein-density-id}
\end{align}

\paragraph{Equation of motion for the metric.}

Now we use \eq{Box-vielbein-density-id} to compute $\Box G^{\mu\nu}$.
For the last term, we observe 
\begin{align}
  \big( \{Z^{\dot\g}, \cE^\mu_{\dot\a}\}\{Z_{\dot\g},\r\} \cE^\nu_{\dot\b}
  + \{Z^{\dot\g}, \cE^\nu_{\dot\a}\}\{Z_{\dot\g},\r\} \cE^\mu_{\dot\b}\big) \eta^{\dot\a\dot\b}
  = \{Z^{\dot\g}, G^{\mu\nu}\}\{Z_{\dot\g},\r\} 
\end{align}
and similarly
\begin{align}
 \{\hat\Theta^{{\dot\a}{\dot\b}},\r\}\cE^\mu_{\dot\b} \tensor{\cE}{_{\dot\a}^\nu} 
  + \{\hat\Theta^{{\dot\b}{\dot\a}},\r\}\cE^\nu_{\dot\a} \tensor{\cE}{_{\dot\b}^\mu} =0 \ .
\end{align}
Thus we obtain
\begin{align}
 \Box G^{\mu\nu} &=  2(m^2 - \r^{-1}\Box\r) G^{\mu\nu} 
  + 2\{Z^{\dot\g},G^{\mu\nu} \}\r^{-1}\{Z_{\dot\g},\r \}  \nn\\
 &\quad  - 2 \tensor{\cE}{_{\dot\a}^\mu} \{\hat\Theta^{{\dot\a}{\dot\b}},\tensor{\cE}{_{\dot\b}^\nu}\} 
   - 2\{\hat\Theta^{{\dot\b}{\dot\a}},\tensor{\cE}{_{\dot\a}^\mu}\} \tensor{\cE}{_{\dot\b}^\nu}
  -2 \{Z_{\dot\g}, \tensor{\cE}{_{\dot\a}^\mu}\} \{Z^{\dot\g}, \tensor{\cE}{^{\dot\a}^\nu}\} \nn\\
 &= 2(m^2 - \r^{-1} \Box\r) G^{\mu\nu} 
   + 2\{Z_{\dot\b},G^{\mu\nu} \}\r^{-1}\{Z_{\dot\b},\r \} \nn\\
  &\quad + 2 \tensor{T}{^\mu^\s^\r} \tensor{\tilde \G}{_\r_\s^\nu}
   + 2 \tensor{T}{^\nu^\s^\r} \tensor{\tilde \G}{_\r_\s^\mu}
   -2 \tensor{\tilde \G}{_{\r}_{\s}^{\mu}}\tensor{\tilde \G}{^{\r}^{\s}^{\nu}}
\end{align}
using
\begin{align}
  \{Z_{\dot\g}, \tensor{\cE}{_{\dot\a}^\mu}\} \{Z^{\dot\g}, \tensor{\cE}{^{\dot\a}^\nu}\} 
   &\sim \g^{\r\r'}\del_\r  \tensor{\cE}{_{\dot\a}^\mu}
                  \del_{\r'} \tensor{\cE}{^{\dot\a}^\nu}   
   =  \g^{\r\r'}\tensor{\tilde \G}{_{\r}_{\dot\a}^{\mu}}\tensor{\tilde \G}{_{{\r'}}^{\dot\a}^{\nu}} \nn\\
   &= \g^{\r\r'}G^{\s\s'}\tensor{\tilde \G}{_{\r}_{\s}^{\mu}}\tensor{\tilde \G}{_{\r'}_{\s'}^{\nu}}
    \nn\\
  \tensor{\cE}{_{\dot\a}^\mu} \{\hat\Theta^{{\dot\a}{\dot\b}},\tensor{\cE}{_{\dot\b}^\nu}\} 
   &\sim \tensor{\cE}{^{\dot\a}^\mu} \tensor{ T}{_{\dot\a}_{\dot\b}^\r}\del_\r \tensor{\cE}{^{\dot\b}^\nu} 
    = - \g^{\mu\mu'}G^{\s\s'}\tensor{T}{_{\mu'}_\s^\r} \tensor{\tilde \G}{_{\r}_{\s'}^\nu}
\end{align}
in the asymptotic regime,
where $\tensor{\tilde \G}{_\r_\s^\nu}$ is the  Weitzenböck connection for the rescaled frame $\cE$
and $\tensor{T}{_{\dot\a}_{\dot\b}^\r}$ is the torsion of basic frame $E$. 
Note that
coordinate indices are raised and lowered with $G^{\mu\nu}$ here.
Now consider  harmonic normal coordinates $y^\mu$ \eq{harm-NC-def}
w.r.t. $G^{\r\s}$  at the point $p$.
Then $\{Z_{\dot\b},G^{\mu\nu} \}$ vanishes at $p$.
Furthermore, the Weitzenb\"ock connection $\tensor{\tilde\G}{_{\r}_{\s}^{\mu}}$ can then be expressed 
in terms of the (con)torsion via 
$\tensor{\tilde \G}{_{\r}_{\s}^{\mu}} = \tensor{\cK}{_{\r}_{\s}^{\mu}}$ \eq{Weitzenbock-Torsion-RiemanNC}.
Thus
\begin{align}
 G_{\mu\mu'}G_{\nu\nu'} \frac 12\Box G^{\mu'\nu'} = (m^2 - \r^{-1} \Box\r)G_{\mu\nu} 
  + \r^2 \big(\tensor{T}{_\mu_\s^\r} \tensor{\cK}{_\r^\s_\nu}
   + \tensor{T}{_\nu_\s^\r} \tensor{\cK}{_\r^\s_\mu}
   - \tensor{\cK}{_{\r}^{\s}_{\mu}}\tensor{\cK}{^{\r}_{\s}_{\nu}}\big) \ .
 \label{Box-G-torsion}
\end{align}
Now in  harmonic NC at $p$, 
the Ricci tensor at $p$ is obtained  as \cite{deturck1981some}
\begin{align}
  \cR_{\mu\nu}[G]|_p &= \frac 12 \Box_G G_{\mu\nu} \ = - G_{\mu\mu'} G_{\nu\nu'} \Box_G G^{\mu'\nu'} \ .
   \label{metric-NC-explicit}
\end{align}
Together with  $\Box = \r^2 \Box_G$ \eq{Box-Box-relation},
we obtain the tensorial equation
\begin{align}
  \cR_{\mu\nu} =  - \tensor{T}{_\mu_\s^\r} \tensor{\cK}{_\r^\s_\nu}
   - \tensor{T}{_\nu_\s^\r} \tensor{\cK}{_\r^\s_\mu}
   + \tensor{\cK}{_{\r}^{\s}_{\mu}}\tensor{\cK}{^{\r}_{\s}_{\nu}}
  + \r^{-2}(-m^2 + \r\Box_G\r) G_{\mu\nu} 
  \label{Ricci-alternative-0}
\end{align}
in the asymptotic regime. 
To see that this coincides with \eq{Ricci-vacuum-1}, we need to compare the contractions of the (con)torsion terms.
One can easily see using \eq{Levi-contorsion-full} and \eq{tilde-T-T-contract} that
\begin{align}
 \tensor{T}{_\mu_\s^\r} \tensor{\cK}{_\r^\s_\nu} + \tensor{T}{_\nu_\s^\r} \tensor{\cK}{_\r^\s_\mu} 
  - \tensor{\cK}{_\r^\s_\mu}  \tensor{\cK}{^\r_\s_\nu}  
  &= \tensor{T}{_\mu_\s^\r}\tensor{K}{_\r^\s_\nu} 
   + \tensor{T}{_\nu_\s^\r} \tensor{K}{_\r^\s_\mu} 
    - \tensor{K}{^\r_\s_\mu}\tensor{K}{_\r^\s_\nu} \nn\\
  &\quad - 2 \r^{-2}\del_\nu \r \del_\mu \r - G_{\mu\nu} \r^{-2}\del^\s \r \del_\s \r    
\end{align}
using
\begin{align}
 \tensor{K}{_\nu^\s_\mu} + \tensor{K}{_\mu^\s_\nu} 
 = \tensor{T}{_\mu^\s_\nu} + \tensor{T}{_\nu^\s_\mu} \ .
\end{align}
Therefore we obtain
\begin{align}
  \cR_{\mu\nu} = - \tensor{T}{_\mu_\s^\r} \tensor{K}{_\r^\s_\nu}
   - \tensor{T}{_\nu_\s^\r} \tensor{K}{_\r^\s_\mu}
   + \tensor{K}{_{\r}^{\s}_{\mu}}\tensor{K}{^{\r}_{\s}_{\nu}}
   + 2 \r^{-2}\del_\nu \r \del_\mu \r 
  + \r^{-2}(-m^2 + \r\Box_G\r + \del^\s\r \del_\s\r) G_{\mu\nu} \ .
  \label{Ricci-alternative}
\end{align}
Expressing the tensorial part in terms of torsion only, one finds using \eq{contraction-T-1} 
\begin{align}
 &\tensor{T}{_\mu_\s^\r} \tensor{K}{_\r^\s_\nu} + \tensor{T}{_\nu_\s^\r} \tensor{K}{_\r^\s_\mu} 
  - \tensor{K}{^\r_\s_\mu}  \tensor{K}{_\r^\s_\nu}  \nn\\
  &= \frac 12\tensor{T}{_\mu^\s^\r} (\tensor{T}{_\nu_\r_\s} + \tensor{T}{_\nu_\s_\r}) 
  + \frac 12\tensor{T}{_\mu^\s^\r} \tensor{T}{_\r_\s_\nu} 
  + \frac 12\tensor{T}{_\nu^\s^\r} \tensor{T}{_\r_\s_\mu}
  + \frac 14 \tensor{T}{^\s^\r_\mu} \tensor{T}{_\r_\s_\nu} 
\end{align}
while for the other form \eq{Ricci-vacuum-1} we need 
\begin{align}
 & - \frac 12\tensor{T}{_\mu_\s^\r} \tensor{T}{_\r^\s_\nu} 
  - \frac 12\tensor{T}{_\nu_\s^\r} \tensor{T}{_\r^\s_\mu}
  -\tensor{K}{_\s^\r_\mu}  \tensor{K}{_\r^\s_\nu}  \nn\\
  &= - \frac 12\tensor{T}{_\mu^\s^\r} \tensor{T}{_\r_\s_\nu} 
  - \frac 12\tensor{T}{_\nu^\s^\r} \tensor{T}{_\r_\s_\mu}
  +\frac 14 \Big(- 2\tensor{T}{_\mu^\s^\r} (\tensor{T}{_\nu_\r_\s} + \tensor{T}{_\nu_\s_\r})
  - \tensor{T}{^\s^\r_\mu}\tensor{T}{_\r_\s_\nu} \Big)
\end{align}
Hence the two expressions for the Ricci tensor are indeed identical.

\subsection{Contractions of the (con)torsion}
\label{sec:torsion-contract}

We need the following rank 2 contractions with the contorsion
\begin{align}
 \tensor{T}{^\mu^\s^\r} \tensor{K}{_\r_\s_\nu} 
 &= \frac 12 
\big( \tensor{T}{^\mu^\s^\r} (\tensor{T}{_\nu_\r_\s} + \tensor{T}{_\nu_\s_\r})
 +\tensor{T}{^\mu^\s^\r} \tensor{T}{_\r_\s_\nu} \big) \nn\\
 %
 -\tensor{K}{^\s^\r^\mu}  \tensor{K}{_\r_\s_\nu}  
 &= \frac 14 \Big(- 2\tensor{T}{^\mu^\s^\r} (\tensor{T}{_\nu_\r_\s} + \tensor{T}{_\nu_\s_\r})
  + \tensor{T}{^\r^\s^\mu}\tensor{T}{_\r_\s_\nu} \Big) \nn\\
 -\tensor{K}{^\r^\s^\mu}  \tensor{K}{_\r_\s_\nu}  
 &= \frac 14 \Big(-2\tensor{T}{^\mu^\r^\s}(\tensor{T}{_\nu_\r_\s} + \tensor{T}{_\nu_\s_\r}) 
  + \tensor{T}{^\s^\r^\mu} \tensor{T}{_\r_\s_\nu}  \Big) 
 \label{contraction-T-1}
 \end{align}
 Note that there are 3 independent rank 2 contractions.
The contraction of the second expression gives  
 \begin{align}
 -\tensor{K}{_\s^\r^\mu}  \tensor{K}{_\r^\s_\mu}  
    &= - \frac 14  \tensor{T}{^\mu^\s^\r} (2\tensor{T}{_\mu_\r_\s} + \tensor{T}{_\mu_\s_\r})
 \end{align}
Therefore
\begin{align}
 - \tensor{T}{_\mu^\r^\s} \tensor{T}{_\r_\s^\mu} 
  -\tensor{K}{_\s^\r^\mu}  \tensor{K}{_\r^\s_\mu}  
 &= - \frac 12 \tensor{T}{^\s^\mu^\r}\tensor{T}{_\mu_\r_\s} 
 - \frac 14 \tensor{T}{^\mu^\s^\r}\tensor{T}{_\mu_\s_\r}  
 \label{scalar-contraction-T}
\end{align}

%
%

\bibliographystyle{JHEP}
\bibliography{papers}


\end{document}